\documentclass[a4paper,onecolumn,superscriptaddress]{quantumarticle}
\pdfoutput=1

\usepackage{amsmath,amssymb}
\usepackage{mathtools}
\usepackage{color}
\usepackage{lineno}
\usepackage{hyperref}
\usepackage{bm}

\usepackage{qcircuit}
\usepackage{algorithm,algorithmic}
\usepackage{braket}
\usepackage{comment}
\usepackage{enumitem}
\usepackage{subfigure}
\usepackage[sorting=none, backend = bibtex, style=numeric-comp,url=false]{biblatex}
\makeatletter
\newenvironment{breakablealgorithm}
{
		\begin{center}
			\refstepcounter{algorithm}
			\hrule height.8pt depth0pt \kern2pt
			\renewcommand{\caption}[2][\relax]{
				{\raggedright\textbf{\ALG@name~\thealgorithm} ##2\par}
				\ifx\relax##1\relax 
				\addcontentsline{loa}{algorithm}{\protect\numberline{\thealgorithm}##2}
				\else
				\addcontentsline{loa}{algorithm}{\protect\numberline{\thealgorithm}##1}
				\fi
				\kern2pt\hrule\kern2pt
			}
		}{
		\kern2pt\hrule\relax
	\end{center}
}
\makeatother

\usepackage{ntheorem}
\newtheorem{theorem}{Theorem}[section]

\newtheorem{lemma}[theorem]{Lemma}

\newtheorem*{proof}{{{Proof.}}}
\newtheorem{definition}{Definition}[section]
\newtheorem{mainresult}{Main Result}

\begin{document}

\title{Beyond Sparsity: Quantum Block Encoding for Dense Matrices via Hierarchically Low Rank Compression} 
	
\author{Kun Tang}
\email{kun.tang@zju.edu.cn}
\affiliation{
	School of Mathematical Sciences, Zhejiang University, Hangzhou, Zhejiang, China, 310058
}

\author{Jun Lai}
\email{laijun6@zju.edu.cn}
\affiliation{
	School of Mathematical Sciences, Zhejiang University, Hangzhou, Zhejiang, China, 310058
}
\affiliation{
    Center for Interdisciplinary Applied Mathematics, Zhejiang University, Hangzhou, Zhejiang, China, 310058
}

\begin{abstract}
While quantum algorithms for solving large scale systems of linear equations offer potentially exponential speedups, their application has largely been confined to sparse matrices. This work extends the scope of these algorithms to a broad class of structured dense matrices arise in potential theory, covariance modeling, and computational physics, namely, hierarchically block separable (HBS) matrices. We develop two distinct methods to make these systems amenable to quantum solvers. The first is a pre-processing approach that transforms the dense matrix into a larger but sparse format. The second is a direct block encoding scheme that recursively constructs the necessary oracles from the HBS structure. We provide a detailed complexity analysis and rigorous error bounds for both methods. Numerical experiments are presented to validate the effectiveness of our approaches.
\end{abstract}

\maketitle

\section{Introduction}
Recent advances \cite{aruteQuantumSupremacyUsing2019} in quantum computing have shown its potential for addressing problems that are difficult for classical computers. Especially, the prospect of exponential speedups has driven the exploration of quantum algorithms across a wide range of scientific and industrial areas. An important application of quantum computing is solving large scale linear systems. The first quantum algorithm for this purpose is the HHL algorithm \cite{harrowQuantumAlgorithmSolving2009}, which provides an exponential speedup for the quantum linear systems problem (QLSP). Subsequently, the quantum singular value transformation (QSVT) framework \cite{gilyenQuantumSingularValue2019} was proposed, which can also be applied to solve linear systems. Readers are referred to \cite{moralesQuantumLinearSystem2025} for a comprehensive review of these methods. 

However, the performance of these algorithms relies on a key assumption: sparsity. The HHL algorithm, for example, depends on the ability to efficiently simulate the Hamiltonian evolution of the system matrix. The efficiency of this simulation, in turn, depends on the matrix's sparsity \cite{childsFirstQuantumSimulation2018}. Similarly, QSVT-based approaches require access to a block encoding of the system matrix. While constructing a block encoding for a sparse matrix is relatively straightforward, doing so for a dense matrix remains a challenge. This sparsity issue becomes severe because, under data access models similar to those used for block encodings, it is possible to construct classical randomized algorithms that match the performance of their quantum counterparts \cite{gilyenImprovedQuantuminspiredAlgorithm2022}. Consequently, better approaches for dense matrices are needed, but existing methods have their own drawbacks. For instance, the quantum algorithms studied in \cite{wossnigQuantumLinearSystem2018} handle dense matrices with only a polynomial speedup instead of an exponential one. 

The goal of this paper is to overcome the difficulty of constructing an efficient block encoding for dense matrices. We do so by exploiting the inherent special structures within certain important classes of dense matrices \cite{nguyenBlockencodingDenseFullrank2022}, including those arising in integral equations, covariance modeling, and computational physics. This structure often arises when matrix entries are defined by a kernel function $K(x,y)$, where interactions between well-separated clusters of points can be approximated by low rank operators. This principle gives rise to structured matrices like hierarchically off-diagonal low rank (HODLR) \cite{laiFastDirectSolver2014} and hierarchically semiseparable (HSS) matrices \cite{xiaFastAlgorithmsHierarchically2010}. HSS matrices employ a nested basis, where the basis at a given hierarchical level is constructed from those of its children.  For more information of these hierarchical structures, see \cite{hackbuschHierarchicalMatricesAlgorithms2015}.

\subsection{Our contributions}
We formally analyze the compressibility of dense matrices under the hierarchically block separable (HBS) framework \cite{hoFastDirectSolver2012}. We quantify how such matrices can be reduced to linear size with uniformly bounded entries. In particular, proxy compression and strong RRQR-based interpolative decomposition enable the construction of hierarchical factors with $O(N)$ computational cost and $O(N)$ memory size. The precise statement are given in Theorem~\ref{thm:hss_sparsification_properties}.

Based on this hierarchical structure, we introduce a quantum algorithm that uses sparsification. We combine the hierarchical factors into a single extended sparse system whose dimension grows only linearly, while its row and column sparsities remain bounded. This procedure makes the system compatible with standard sparse matrix block encoding oracles.
\begin{mainresult}\label{thm:informal_sparse_blockencoding}
	Let $A_{\mathrm{sp}}$ be the extended sparse system derived from an HBS approximation of a dense matrix $A$. Then $A_{\mathrm{sp}}$ has bounded row and column sparsities and contains only $O(N)$ nonzero entries. Consequently, one can implement a block encoding of $A_{\mathrm{sp}}$ with subnormalization factor bounded by a constant (independent of $N$), using oracles that access only $O(N)$ stored nonzeros.
\end{mainresult}
A detailed formulation in the oracle model is given in Theorem~\ref{thm:sparsification_block}.

We also analyze the side effects introduced by extended sparsification. Our analysis addresses the probability loss from auxiliary variables, the conditional number of the system via a regularized formulation and the propagation of the HBS approximation error to the solution. The precise statement and constants are given in Theorem \ref{thm:sparsification_regularization_error}. 

Finally, we develop a direct recursive construction of a block encoding for the HBS approximation itself. By using QMM to multiply block encodings along the hierarchy and LCU to add near field blocks level by level, this construction yields upper bounds for subnormalization growth and error accumulation.
\begin{mainresult}\label{thm:informal_direct_recursive}
	Given block encodings of the HBS factors $\{D_\ell,L_\ell,R_\ell\}$ with bounded subnormalization factors, there is a recursive procedure (based on QMM and LCU) that produces a block encoding of the full HBS approximation. The resulting subnormalization factor and overall precision loss admit upper bounds determined by the hierarchy depth and the factor-wise subnormalizations.
\end{mainresult}
The complete quantitative statement is Theorem~\ref{thm:recursive_be_analysis} in Section~\ref{se:direct_be}. 
A detailed derivation of the subnormalization and error bounds is provided immediately after the theorem.

\subsection{Related works and comparison}

The concept of using hierarchical matrix structures for quantum block encoding was recently explored by Nguyen et al. \cite{nguyenBlockencodingDenseFullrank2022}. They introduced a method for block encoding dense kernel matrices by splitting them into block sparse hierarchical levels and using a linear combination of unitaries (LCU). Their approach achieves optimal subnormalization factors for kernels with specific analytical decay properties. However, because it evaluates entries across all dense admissible blocks at each level, the strategy's data access cost scales as $O(N^2)$ for general matrices.

Our work takes a different approach by using the HBS format. We employ nested interpolative decompositions and proxy surfaces to compress the matrix into sparse factors with only $O(N)$ nonzeros. Furthermore, our method removes the need for an analytical kernel representation, making it applicable to a broader range of structured or perturbed matrices.

The remainder of this paper is organized as follows. In Section \ref{se:preliminary}, we review the basic formulations of quantum linear system algorithms and show the obstacles in constructing block encodings for dense matrices. Section \ref{se:HBS_algorithm} introduces HBS matrices and presents a modified algorithm adapted for the subsequential block encoding in quantum setting. In Section \ref{se:sparsification_quantum}, we propose a sparsification based quantum approach, which uses an extended sparsification to combine the HBS matrix components into a single but slightly larger system. As an alternative, Section \ref{se:direct_be} describes a direct method for constructing the block encoding from the hierarchical matrix structure. Section \ref{se:numerical} presents numerical results from experiments on integral equations to evaluate the performance of our methods based on the classical computers. Finally, Section \ref{se:conclusions} provides concluding remarks and discusses potential directions for future work.

\section{Quantum linear system algorithms}\label{se:preliminary}

To solve linear systems in the quantum setting, we first review the definition of the quantum linear system problem (QLSP) \cite{childsQuantumAlgorithmSystems2017}. In the following presentation, if not explicitly mentioned, we always denote $\|\cdot\|$ the $2$-norm.
\begin{definition}[QLSP]\label{def:QLSP}
    Consider an $N\times N$ Hermitian matrix $A$, which has condition number $\kappa$ and is normalized such that $\Vert A\Vert =1$. Let $\vec{b}$ be an $N$-dimensional vector and $\vec{x}:=A^{-1}\vec{b}$. The corresponding quantum states are defined as
    \[\ket{b}:=\frac{\sum_ib_i\ket{i}}{\Vert \sum_ib_i\ket{i}\Vert},\ \ket{x}:=\frac{\sum_ix_i\ket{i}}{\Vert \sum_ix_i\ket{i}\Vert}.\]
    Given access to a procedure $\mathcal{P}_A$ for the entries of $A$ and a state preparation procedure $\mathcal{P}_b$ that creates $\ket{b}$ in time complexity $\mathcal{O}({\rm{poly}}(\log N))$, the goal is to output a state $\ket{\widetilde{x}}$ satisfying $\Vert \ket{x}-\ket{\widetilde{x}}\Vert \leq \varepsilon$. The algorithm should succeed with probability $\Omega(1)$ and provide a flag indicating success.
\end{definition}

Several algorithms have been developed to solve the QLSP. The most well known is the HHL algorithm, which uses phase estimation and Hamiltonian simulation \cite{harrowQuantumAlgorithmSolving2009}. Another one is QSVT \cite{gilyenQuantumSingularValue2019}, which provides a general framework for applying polynomial functions to a matrix. Especially, the QSVT uses a polynomial to approximate the inverse function $1/x$, and apply the pseudoinverse $A^+$ to the state $\ket{b}$. However, a prerequisite for these approaches is the ability to embed the matrix $A$ into a larger unitary operator, a procedure known as block encoding. 

\begin{definition}[Block Encoding]\label{def:block_encoding}
    An $(n+a)$-qubit unitary $U$ is an $(\alpha, a, \varepsilon)$-block-encoding of a $2^n\times 2^n$ matrix $A$ if it satisfies:
    \[\Vert A-\alpha(\bra{0}^{\otimes a}\otimes I_{2^n})U(\ket{0}^{\otimes a}\otimes I_{2^n})\Vert \leq\varepsilon,\]
    where $\alpha \ge \|A\|$ is the subnormalization factor, $a$ is the number of ancillary qubits, and $\varepsilon>0$ is a sufficiently small precision number.
\end{definition}

The construction of block encodings can be achieved through several established methods, particularly for matrices with specific sparsity structures. We call a matrix $A$ is $s$-row-sparse ($s$-column-sparse) if each row (column) of $A$ has at most $s$ nonzero elements. The following lemma, adapted from Gilyén et al. \cite{gilyenQuantumSingularValue2019} and summarized in Nguyen et al. \cite{nguyenBlockencodingDenseFullrank2022}, provides an example for oracle accessible sparse matrices.
\begin{lemma}[Block encoding \cite{nguyenBlockencodingDenseFullrank2022}]\label{apth:block_encoding}
    Let $A\in\mathbb{C}^{2^n\times 2^n}$ be a matrix that is $s_r$-row-sparse and $s_c$-column-sparse. Suppose we are given access to the following $(n+1)$-qubit oracles:
    \[\mathcal{O}_r:\ket{i}\ket{k}\to\ket{i}\ket{r_{ik}}\quad \text{for } 0\leq i<2^n,0\leq k<s_r,\]
    \[\mathcal{O}_c:\ket{l}\ket{j}\to\ket{c_{lj}}\ket{j}\quad \text{for } 0\leq j<2^n,0\leq l<s_c,\]
    where $r_{ik}$ is the column index of the $k$-th nonzero entry in the $i$-th row of $A$, and $c_{lj}$ is the row index of the $l$-th nonzero entry in the $j$-th column of $A$. If there are fewer than $k$ nonzero entries, the oracle $\mathcal{O}_r$ maps $\ket{i}\ket{k}$ to $\ket{i}\ket{2^n+k}$, and $\mathcal{O}_c$ has the same action. Furthermore, let $\hat{a}\geq \max_{i,j}\vert a_{ij}\vert$, and suppose we have an oracle for the matrix entries:
    \[\mathcal{O}_A:\ket{i}\ket{j}\ket{0}^{\otimes b}\to \ket{i}\ket{j}\ket{\tilde{a}_{ij}},\quad {\rm{for }}\ 0\leq i,j<2^n,\]
    where $\tilde{a}_{ij}$ is the $b$-qubit fixed-point representation of $a_{ij}/\hat{a}$. Then, one can implement an $(\hat{a}\sqrt{s_rs_c},n+3,\varepsilon)$-block-encoding of $A$ with a single query to $\mathcal{O}_r$ and $\mathcal{O}_c$, two queries to $\mathcal{O}_A$, and $O(n+{\rm polylog}(\frac{\hat{a}\sqrt{s_r s_c}}{\varepsilon}))$ additional one and two qubit gates, using $O(b,{\rm polylog}(\frac{\hat{a}\sqrt{s_r s_c}}{\varepsilon}))$ ancilla qubits.
\end{lemma}

Lemma \ref{apth:block_encoding} establishes that the efficiency of the block encoding, specifically the subnormalization factor $\alpha$, is determined by the geometric mean of the row and column sparsities. While this construction is highly efficient for sparse matrices, its limitations become apparent for dense matrices. For example, if $s_r = s_c = 2^n$, the subnormalization factor $\alpha = \hat{a}\sqrt{s_r s_c}$ scales as $\hat{a} \cdot 2^n$. Since the complexity of algorithms using this block encoding typically scales polynomially with $\alpha$, this exponential growth in $n$ poses a significant drawback. We address this problem by proposing and analyzing two methods. The first uses an extended sparsification technique to convert the dense matrix into a sparse one, followed by an analysis of the cost to construct a block encoding for the resulting system. The second is a direct method for block encoding the hierarchical matrix structure itself, for which we analyze the complexity and the required subnormalization factor.

In the complexity analysis that follow, we account for the number of nonzero elements. This number provides a more realistic measure of resource requirements than simply using the matrix dimensions. The cost of implementing the oracles, such as $\mathcal{O}_r,\mathcal{O}_c$ and $\mathcal{O}_A$ in Lemma \ref{apth:block_encoding}, is directly related to the gate count needed to access the value and location of each nonzero entry. In particular, \cite{zhangCircuitComplexityQuantum2024} shows that the gate count for constructing these oracles is directly proportional to the number of nonzero elements in the matrix representation. Therefore, to implement our proposed methods effectively, we first require a suitable way to compress the dense matrix into a format with a manageable number of entries. The next section introduces the hierarchically block separable (HBS) format \cite{hoFastSemidirectLeast2014} that serves as the foundation for our analysis.

\section{Hierarchically block separable matrices}\label{se:HBS_algorithm}
Solving large scale linear systems is a classical problem in scientific computing. While standard iterative methods like Bi-CGSTAB and GMRES are often applied, their computational efficiency is matrix dependent \cite{doi:10.1137/1.9781421407944}, which can lead to poor performance for the dense matrices that arise in many physical problems. An effective strategy is to exploit the specific structure hidden behind these matrices. In particular, matrices from the discretization of integral equations are not arbitrarily dense but possess a sparse structure. This property stems from the smoothness of the underlying kernel $K(x,y)$ for well separated points $x$ and $y$ in 2D or 3D. Important kernels in potential theory and wave physics that exhibit this behavior include:
\begin{enumerate}
    \item $\log \vert x-y \vert$, for 2D potential problems,
    \item $H^{(1)}_0(\kappa\vert x-y\vert)$,  for 2D wave scattering,
    \item $1/\vert x-y \vert$, for 3D potential problems,
    \item $e^{i\kappa|x-y|}/\vert x-y \vert$, for 3D wave scattering,
\end{enumerate}
where $\kappa>0$ is the wavenumber, and $H^{(1)}_0$ is the first kind Hankel function of order zero.

The HBS matrix format provides an efficient way to represent this structure \cite{hoFastDirectSolver2012, mindenRecursiveSkeletonizationFactorization2017}.  To define this structure formally, consider a dense matrix $A\in\mathbb{C}^{N\times N}$. Let its row and column index sets be $I = \{1, \dots, N\}$ and $J = \{1, \dots, N\}$, respectively. We begin by partitioning the row index set $I$ into $p$ disjoint blocks, $\{I_j\}_{j=1}^p$. The $j$-th row block, $I_j$, contains $n_j = |I_j|$ indices and is defined as:
\[I_j = \left\{ p_{j-1} + 1, p_{j-1}+2,\dots, p_j \right\}.\]
where $p_j = \sum_{k=1}^j n_k$ represents the cumulative size of the first $j$ blocks. By convention, $p_0=0$.
Similarly, the column index set $J$ is partitioned into $p$ blocks, $\{J_k\}_{k=1}^p$, where the $k$-th block, $J_k$, has a size of $n'_k = |J_k|$. Using this partitioning, the linear system $Ax=b$ can be expressed in the block form:
\[\sum_{j=1}^{p} A_{ij}x_j = b_i, \quad \text{for } i=1, \dots, p,\]
where $x_j \in \mathbb{C}^{n'_j}$ and $b_i \in \mathbb{C}^{n_i}$ are the corresponding blocks of the vectors $x$ and $b$, and $A_{ij} \in \mathbb{C}^{n_i \times n'_j}$ is the submatrix at the intersection of row block $I_i$ and column block $J_j$. Given an error tolerance $\varepsilon$, a matrix $A$ is defined as block separable if each of its off-diagonal blocks, $A_{ij}$ for $i \neq j$, admits a low-rank factorization of the form:
\[A_{ij} \approx A_{ij,\varepsilon} =  L_i S_{ij} R_j.\]
Here, $L_i \in \mathbb{C}^{n_i \times k_i^r}$ and $R_j \in \mathbb{C}^{k_j^c \times n_j'}$ are basis matrices, and $S_{ij} \in \mathbb{C}^{k_i^r \times k_j^c}$. The ranks $k_i^r$ and $k_j^c$ are assumed to be much more smaller than the block dimensions, i.e., $k_i^r \ll n_i$ and $k_j^c \ll n_j'$.

This sparse representation can be constructed using the interpolative decomposition (ID) \cite{halkoFindingStructureRandomness2011}, which is defined as follows:
\begin{definition}\label{def:ID}
    Let $A\in \mathbb{C}^{m\times n}$ be a matrix with a target rank $k < \min(m, n)$. A rank-$k$ ID of $A$ is a factorization of the form $A \approx A_\varepsilon = BP$ under a given error tolerance $\varepsilon$, where $B \in \mathbb{C}^{m\times k}$ is a skeleton matrix whose columns are a specific subset of $k$ columns from $A$. The matrix $P \in \mathbb{C}^{k\times n}$ is an interpolation matrix that contains the $k\times k$ identity matrix as a submatrix.
\end{definition}
Efficient randomized algorithms can compute the ID with a typical complexity of $O(mn\log(k) + k^2n)$ \cite{halkoFindingStructureRandomness2011}. ID is preferred over other low-rank methods like the singular value decomposition (SVD) for several reasons. Firstly, the skeleton matrix $B$ consists of actual columns from $A$, which preserves physical structure. Secondly, the interpolation matrix $P$ has a provably bounded norm, ensuring numerical stability. We use the ID algorithm to approximate the off-diagonal blocks $A_{ij}$ (for $i\neq j$) as $A_{ij} \approx L_iS_{ij}R_j$, where $S_{ij}\in\mathbb{C}^{k_i^r\times k_j^c}$ with ranks satisfying $k_i^r\ll n_i$ and $k_j^c\ll n_j$.

This recursive compression algorithm, illustrated in Figure \ref{fig:illustrate2}, constructs the hierarchical factors ($D_l, L_l, R_l$) level by level. Let $S_0 = A$ be the initial matrix. At each level $l$ of the hierarchy, the algorithm processes the current matrix $S_{l-1}$, which is partitioned into blocks based on the spatial tree. The blocks corresponding to near field, where the target $x$ and source $y$ are close, are typically dense, and these blocks are extracted and collectively form the diagonal factor $D_l$. The remaining matrix, which we denote as $B_l = S_{l-1} - D_l$, contains only the far field interactions. These far field blocks are presumed to be numerically low-rank. They are approximated using ID to identify a small set of skeleton indices and to construct basis matrices. To perform this efficiently, all far-field blocks within a given row partition share a common basis matrix, and similarly, all blocks within a column partition share a common column basis. The computational cost of these ID procedures can be significantly reduced by using proxy approximation, which replaces global computations with local ones involving a small set of sources \cite{xingInterpolativeDecompositionProxy2020}; a more detailed description is provided in Section \ref{apse:1.1}. These bases are then assembled into the block-diagonal factors $L_l$ and $R_l$. The far field part is thus approximated as $B_l \approx L_l S_l R_l$, where $S_l$ is the skeleton matrix. This new matrix $S_l$ is passed to the next level, $l+1$, for the same decomposition procedure. This process is repeated until a predefined recursion depth $\lambda$ is reached, and the final remaining skeleton matrix is set as $D_{\lambda+1}$.

\begin{figure*}[htbp]
\centering
\includegraphics[width=\textwidth]{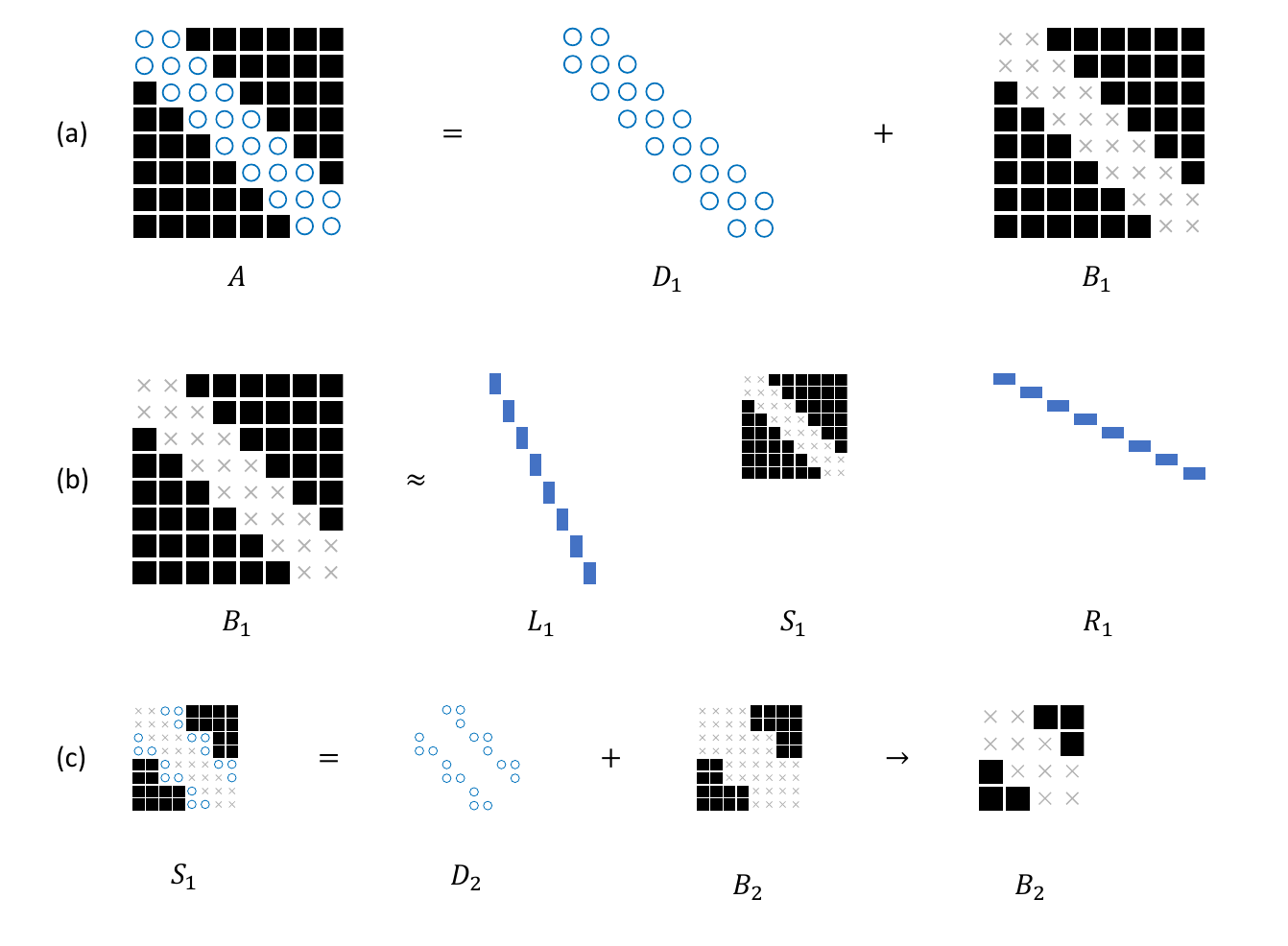}
\caption{\label{fig:illustrate2} This figure illustrates the algorithmic process of the recursive HBS factorization. At level 1, (a) the matrix $A$ is decomposed into its near-field blocks ($D_1$) and far-field remainder ($B_1$). (b) The far-field part is approximated as $B_1 \approx L_1 S_1 R_1$ using ID.  (c) The process is recursively applied to $S_1$, and the indices of the resulting far-field part $B_2$ are reordered for the subsequent level.}
\end{figure*}

The recursive application of this compression is guided by a partition tree. A standard approach is to recursively bisect the row and column index sets, as shown in Figure \ref{fig:illustrate1}. However, this index-based partitioning has a major limitation in higher dimensions. In 1D, nearby indices often correspond to physically close points, but this correspondence breaks down in 2D or 3D. Geometrically close points can have widely separated indices, which means their interaction blocks will appear off-diagonal. Such blocks are typically high-rank and cannot be compressed effectively \cite{hoFastSemidirectLeast2014}.

\begin{figure*}[htbp]
\centering
\includegraphics[width=\textwidth]{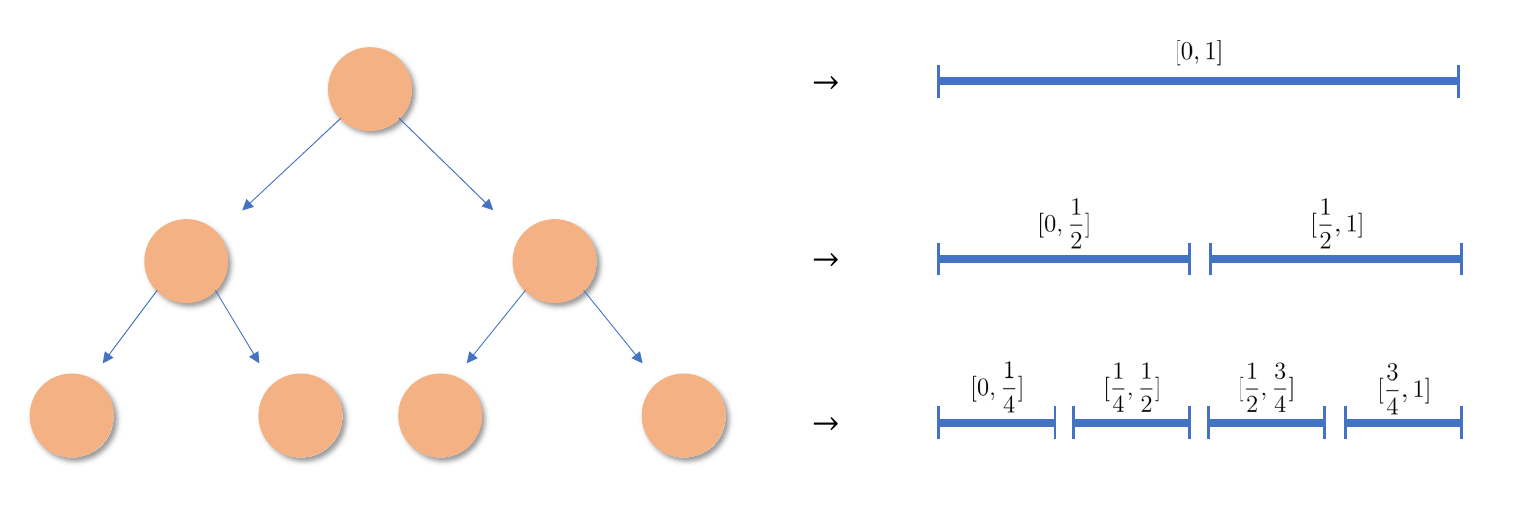}
\caption{\label{fig:illustrate1} An index tree for the interval $[0,1]$. The interval is first partitioned into $[0,\frac{1}{2}]$ and $[\frac{1}{2},1]$, and the two resulting nodes are then partitioned further. Each node at each level corresponds to an interval.}
\end{figure*}

To address this limitation, based on the HBS algorithm \cite{hoFastDirectSolver2012}, a more robust strategy uses the geometric coordinates of the underlying data points, which are typically available for matrices defined by a kernel function \cite{coulierInverseFastMultipole2017}. Instead of partitioning indices based on their ordering, the hierarchy is constructed directly from these spatial coordinates, for instance, by employing a quadtree in 2D or an octree in 3D. For matrices that are not strictly generated by a kernel function, precise coordinates are not strictly required for this partitioning; approximate coordinates are sufficient. For example, if the matrix represents a kernel-based system subject to a small perturbation, the coordinates associated with the original kernel function can still be effectively used. This hierarchical structure ensures that matrix blocks corresponding to well-separated spatial clusters are low rank and compressible, whereas blocks corresponding to adjacent points are generally treated as dense diagonal blocks at a given level. The procedure is summarized in Algorithm \ref{alg1} and illustrated in Figure \ref{fig:illustrate2}.

\begin{breakablealgorithm}
    \caption{Construction of HBS Factors via Spatial Partitioning}
    \label{alg1}
    \begin{algorithmic}[1]
        \REQUIRE 
        Matrix $A \in \mathbb{C}^{N \times N}$; 
        A spatial tree $\mathcal{T}$ partitioning the indices $\{1, \dots, N\}$;
        Recursion depth $\lambda$;
        Compression tolerance $\varepsilon$.
        
        \ENSURE 
        Diagonal blocks $D=\{D_1, \dots, D_\lambda, D_{\lambda+1}\}$; 
        Row basis matrices $L=\{L_1, \dots, L_\lambda\}$;
        Column basis matrices $R=\{R_1, \dots, R_\lambda\}$.

        \STATE $S \leftarrow A$
        \STATE $I_{\rm{active}} \leftarrow \{1, \dots, N\}$ \textit{// Set of active row indices for the current level}
        \STATE $J_{\rm{active}} \leftarrow \{1, \dots, N\}$ \textit{// Set of active column indices for the current level}
        
        \FOR{$k=1$ to $\lambda$}
            \STATE Partition $I_{active}$ and $J_{active}$ into blocks $\{\mathcal{I}_p\}$ and $\{\mathcal{J}_q\}$ based on level $k$ of the tree $\mathcal{T}$.
            
            \STATE \textit{// Extract near-field blocks to form $D_k$}
            \STATE $D_k \leftarrow$ Submatrix of $S$ formed by all near-field blocks at level $k$.
            \STATE $S\leftarrow S-D_k$ \textit{// Remove the diagonal elements}
            
            \STATE \textit{// Compress blocks to find skeletons and bases}
            \STATE $I_{\rm{skel}} \leftarrow \emptyset$, $J_{\rm{skel}} \leftarrow \emptyset$, $U_{\rm{list}} \leftarrow \{\}$, $V_{\rm{list}} \leftarrow \{\}$
            \FOR{each row block $\mathcal{I}_p$}
                \STATE Apply ID to the submatrix $S(\mathcal{I}_p, \mathcal{J}^{(p)}_{\rm{far}})$ with tolerance $\varepsilon$.
                \STATE Let $\mathcal{I}'_p \subseteq \mathcal{I}_p$ be the skeleton row indices and $U_p$ be the corresponding basis matrix.
                \STATE $I_{\rm{skel}} \leftarrow I_{\rm{skel}} \cup \mathcal{I}'_p$
                \STATE Append $U_p$ to $U_{\rm list}$.
            \ENDFOR
            \FOR{each column block $\mathcal{J}_q$}
                \STATE Apply ID to the submatrix $S(\mathcal{I}^{(q)}_{\rm{far}}, \mathcal{J}_{q})$ with tolerance $\varepsilon$.
                \STATE Let $\mathcal{J}'_q \subseteq \mathcal{J}_q$ be the skeleton column indices and $V_q$ be the corresponding basis matrix.
                \STATE $J_{\rm{skel}} \leftarrow J_{\rm{skel}} \cup \mathcal{J}'_q$
                \STATE Append $V_q$ to $V_{\rm{list}}$.
            \ENDFOR

            \STATE \textit{// Assemble factors and update for the next level}
            \STATE $L_k \leftarrow \text{blkdiag}(U_p \in U_{\rm{list}})$ \textit{// Assemble row bases into a block-diagonal matrix}
            \STATE $R_k \leftarrow \text{blkdiag}(V_q \in V_{\rm{list}})$ \textit{// Assemble column bases into a block-diagonal matrix}
            \STATE $S \leftarrow S(I_{\rm{skel}}, J_{\rm{skel}})$ \textit{// Update S to be the skeleton matrix}
            \STATE $I_{\rm{active}} \leftarrow I_{\rm{skel}}$, $J_{\rm{active}} \leftarrow J_{\rm{skel}}$
        \ENDFOR
        
        \STATE $D_{\lambda+1} \leftarrow S$ \textit{// The final matrix is the last diagonal block}
    \end{algorithmic}
\end{breakablealgorithm}

This process yields the telescoping representation
\begin{equation}\label{eq:recursive}
    A\approx D_1+L_1[D_2+L_2(\cdots D_\lambda+L_\lambda D_{\lambda+1} R_\lambda\cdots)R_2]R_1,
\end{equation}
where $\lambda$ is the maximum depth of the hierarchy. The recursion is typically terminated when the matrix blocks at the final level reach a predefined size $M$. It determines the depth as $\lambda=O(\log N)$. 

To analyze the complexity and sparsity of the algorithm, denote each level $l=1,\dots,\lambda$, where $l=1$ corresponds to the finest partition. Let $p_l$ be the number of blocks along one dimension, $n_l$ be the uncompressed block size, and $k_l$ be the compressed block size (i.e., the rank). For simplicity, we assume these parameters are uniform for all blocks at a given level. The spatial dimension is denoted by $d$. Our analysis is based on the following assumptions:
\begin{enumerate}[label=(\roman*)]
    \item The total matrix dimension is $N = p_1 n_1$. We assume the finest-level blocks contain a constant number of points $M$, so $n_1 = O(1)$, which implies $p_1 = O(N)$. For simplicity, we set $p_1=2^{d\lambda}$.
    \item At each level of the hierarchy, a parent block is divided into $2^d$ children. This means the number of blocks is divided by $2^d$ at each coarser level, so $p_l = p_1/2^{d(l-1)}$. Since $p_\lambda = O(1)$, the total number of levels is $\lambda=\lfloor \log_2(N/M)/d\rfloor+1= O(\log N)$.
    \item The size of the matrix at level $l$ is determined by the number of skeleton points from level $l-1$. This gives the relation $p_l n_l = p_{l-1} k_{l-1}$, so $n_l = 2^dk_{l-1}$.
    \item The ranks of the off-diagonal blocks corresponding to well-separated points are bounded by a constant $r$ that depends on the desired precision but not on the matrix size $N$. We therefore assume $k_l = r$ for all levels $l$.\label{assump}
\end{enumerate}

Assumption \ref{assump} is justified for matrices generated by smooth kernel functions $K(x,y)$. The Taylor expansion of such a kernel shows that for well-separated domains, it can be approximated by a finite sum of separable functions:
\begin{equation}
\begin{aligned}
    K(x,y)&\approx\sum\limits_{\vert {\bm{\beta}}\vert=0}^{r-1}\frac{(y-y_0)^{\bm{\beta}}}{{\bm{\beta}} !}D^{\bm{\beta}}_2K(x,y_0)\\
    &\approx\sum\limits_{\vert {\bm{\alpha}}\vert=0}^{r-1}\sum\limits_{\vert {\bm{\beta}}\vert=0}^{r-1}\frac{(x-x_0)^{\bm{\alpha}}}{{\bm{\alpha}} !}\frac{(y-y_0)^{\bm{\beta}}}{{\bm{\beta}} !}D_1^{\bm{\alpha}} D^{\bm{\beta}}_2K(x_0,y_0).
\end{aligned}
\end{equation}
Truncating this expansion yields a low-rank approximation where the rank $r$ is determined by the required accuracy. This is the same underlying principle that enables the efficiency of algorithms like the fast multipole method \cite{bebendorfHierarchicalMatricesMeans2008}.

The classical procedure for constructing the decomposition of matrix $A$ leads to the following result regarding its computational cost and sparsity. 
\begin{theorem}\label{thm:hss_sparsification_properties}
Let $A \in \mathbb{C}^{N \times N}$ be a matrix that has a hierarchical representation of depth $\lambda$. If this representation is built using ID and a strong rank-revealing QR (RRQR) algorithm with parameter $f$, then the construction has a computational cost of 
\begin{equation}
    T \sim 2^{d+1}(mr\log r+r^3)(2^{d(\lambda+1)}-2^d)/(2^d-1)\sim O(N),
\end{equation}
where $r$ is the maximum rank of the decompositions and $m$ is the proxy set size. The entries of the resulting matrices are bounded by a constant $c$. The matrices are sparse, containing a total of 
\begin{widetext}
\begin{equation}
s = 3^dn_1N+2^d[(2^{d+1}r^2-2r^2+4r)(2^{d\lambda}-1)+4^d(6^d-3^d)r^2(2^{d\lambda}-1)]/(2^d-1) \sim O(N)
\end{equation}
\end{widetext}
nonzero elements.
\end{theorem}
\begin{proof}
	The statement follows from three components in hierarchical matrix compression: 
	\begin{enumerate}[label=(\arabic*)]
	\item Proxy approximation to localize the ID computations (Section~\ref{apse:1.1}).
	\item A level-by-level count of nonzeros and ID costs under uniformly bounded ranks (Section~\ref{apse:1.2}).
	\item Entrywise control of the interpolation matrices produced by strong RRQR (Section~\ref{apse:1.3}). 
\end{enumerate}
\end{proof}

Theorem \ref{thm:hss_sparsification_properties} states that both the classical preparation cost and the number of nonzero elements in the resulting matrices scale linearly with the matrix dimension $N$. Additionally, the elements of these matrices are uniformly bounded. These properties are essential for constructing sparse matrix oracles and analyzing their quantum costs later in this paper.

In the following subsections, we detail the hierarchical compression procedure that leads to these bounds. We begin by introducing the proxy approximation, which reduces the global compression step to a local one (Section~\ref{apse:1.1}). Next, we calculate the nonzero count and runtime (Section~\ref{apse:1.2}), and finally, we explain how the strong RRQR algorithm ensures bounded matrix entries (Section~\ref{apse:1.3}).

\subsection{Proxy Approximation}\label{apse:1.1}
A major computational cost in the algorithm comes from the compression step. For each block, finding the basis vectors requires information from all other blocks in the same row or column, which is a global operation. If the matrix structure is arbitrary, this is unavoidable. However, for many matrices that come from physical problems, this global work can be replaced with a much cheaper local computation. The main idea, known as proxy approximation \cite{xingInterpolativeDecompositionProxy2020}, is to replace the influence of all distant parts of the system with an equivalent set of sources on a small, artificial surface.

As shown in Figure \ref{fig:proxy}, to compress a block corresponding to a cluster of points $X_0$, we no longer need to consider all other clusters $Y_0$. Instead, we only need to account for its immediate neighbors and a nearby ``proxy surface" $\Gamma_0$ that separates the local region from the rest of the domain. The field generated by all distant sources can be replicated by an equivalent distribution on this proxy surface. As a result, the compression becomes a local step. The basis for a block is found by interacting it only with its neighbors and the proxy surface, not the entire system. Since the number of points needed for the proxy surface is much smaller and does not depend on the total problem size, this approach can greatly reduce the computational cost.
\begin{figure*}
	\centering
	\includegraphics[width=0.35\linewidth]{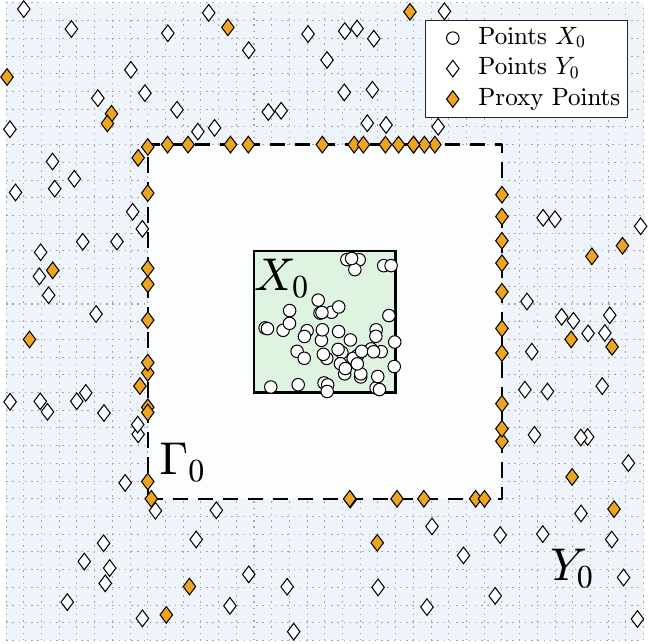}
	\caption{An illustration of proxy approximation. }
	\label{fig:proxy}
\end{figure*}

\subsection{Computational Complexity}\label{apse:1.2}
In this section, we analyze the storage and computational complexity of constructing the hierarchical representation. The key metrics are determined by the properties of the factor matrices $\{D_l, L_l, R_l\}$. The total number of nonzero elements in the extended sparse matrix $A_{\text{sp}}$ is of the same order as the sum of nonzeros in these factors. Therefore, an analysis of the storage for these factors provides an effective bound on the overall storage cost of $A_{\text{sp}}$ and allows us to verify that its row and column sparsity is bounded by a constant.

We estimate the total number of nonzero elements, $s$, by summing the contributions from all blocks. The total number of nonzeros in the diagonal blocks is $3^d n_1^2 p_1$ for $l=1$ and $(6^d-3^d) n_l^2 p_l$ for $l>1$. The interpolation matrices and other sparse components also contribute to $s$. Summing these quantities over all levels gives the total number of nonzero elements:
\begin{widetext}
	\begin{equation}
		\begin{aligned}
			s&=3^dn_1^2p_1+\sum\limits_{l=1}^\lambda (2k_lp_l+2k_l(n_l-k_l+1)p_l+(6^d-3^d)n_l^2p_l)\\
			&=3^dn_1N+(2^{d+1}r^2-2r^2+4r)(2^{d(\lambda+1)}-2^d)/(2^d-1)+4^d(6^d-3^d)r^2(2^{d(\lambda+1)}-2^d)/(2^d-1)\\
			&\sim O(N).
		\end{aligned}
	\end{equation}
\end{widetext}
This linear scaling shows that the extended sparsification reduces the storage requirement from ${O}(N^2)$ for the original dense matrix to ${O}(N)$.

The computational cost is dominated by the construction of the hierarchical representation, specifically the ID performed at each level. The cost of a rank-$k$ ID on an $m \times n$ matrix is typically ${O}(mn\log k+k^2n)$. Summing this over all blocks and levels gives the total complexity for the compression step:
\begin{equation}
	T\sim\sum\limits_{l=1}^\lambda 2p_l(m_ln_l\log k_l+k_l^2n_l).
\end{equation}

The key to achieving linear complexity is the use of proxy points, as in Section \ref{apse:1.1}. With this approach, the ID for a given block is not performed against all other blocks in the system. Instead, it is only performed against a much smaller matrix formed by its near-neighbors and the proxy points. This reduces the dimension $m_l$ in the cost formula from ${O}(N)$ to a small constant, $m$, that depends only on the local geometry. Since the block sizes $n_l$ and ranks $k_l$ are also bounded by constants ($2^dr$ and $r$), the cost to compress a single block becomes constant. The total cost is then simplified to:
\begin{widetext}
	\begin{equation}
		T\sim2(2^dmr\log r+2^{d}r^3)\sum\limits_{l=1}^\lambda p_l=2^{d+1}(mr\log r+r^3)(2^{d(\lambda+1)}-2^d)/(2^d-1)\sim O(N).
	\end{equation}
\end{widetext}
Therefore, the total computational cost for constructing the hierarchical representation scales linearly with the problem size, $N$.

\subsection{Element Bound}\label{apse:1.3}
For the subsequent analysis, it is important that the entries of the matrices in the HBS representation are bounded. This bound can be guaranteed by constructing the ID using a strong rank-revealing QR (RRQR) algorithm \cite{guEfficientAlgorithmsComputing1996}. We first formalize the definition of this factorization.
\begin{definition}[Strong RRQR]\label{def:RRQR_def}
	Given a matrix $M\in\mathbb{R}^{m\times n}$, we consider partial QR factorizations of the form 
	\begin{equation*}
		M\Pi=QR\equiv Q\begin{pmatrix}
			A_k&B_k\\
			&C_k
		\end{pmatrix},
	\end{equation*}
	where $Q\in\mathbb{R}^{m\times m}$ is orthogonal, $A_k\in\mathbb{R}^{k\times k}$ is upper triangular with nonnegative diagonal elements, $B_k\in\mathbb{R}^{k\times (n-k)}, C_k\in\mathbb{R}^{(m-k)\times (n-k)}$, and $\Pi\in\mathbb{R}^{n\times n}$ is a permutation matrix.
	We call the factorization a strong RRQR if it satisfies 
	\[\sigma_i(A_k)\geq\frac{\sigma_i(M)}{q_1(k,n)}\quad {\rm {and}}\quad \sigma_j(C_k)\leq \sigma_{k+j}(M)q_1(k,n)\]
	and
	\[\vert (A_k^{-1}B_k)_{i,j}\vert\leq q_2(k,n),\]
	for $1\leq i\leq k$ and $1\leq j\leq n-k$, where $q_1(k,n)$ and $q_2(k,n)$ are functions bounded by low-degree polynomials in $k$ and $n$.
\end{definition}
Gu and Eisenstat \cite{guEfficientAlgorithmsComputing1996} proved that such a factorization can be computed.
\begin{theorem}\label{apth:RRQR}
	For any given parameter $f\geq 1$, there exists an algorithm that computes a strong RRQR factorization satisfying the conditions in Definition \ref{def:RRQR_def} with
	\[q_1(k,n)=\sqrt{1+f^2k(n-k)}\quad {\rm {and}}\quad q_2(k,n)=f.\]
\end{theorem}
Theorem \ref{apth:RRQR} ensures that the entries of the interpolation matrix, which corresponds to $A_k^{-1}B_k$ in the factorization, are bounded by a predefined constant $f$. The entries of the diagonal blocks $D_l$ are taken directly from the original matrix $A$ and are therefore also assumed to be bounded. In conclusion, the elements in $A_{\text{sp}}$ are bounded by a constant $c_{\rm sp}$.

\section{A Sparsification Based Quantum Approach}\label{se:sparsification_quantum}

The decomposition algorithm from Section \ref{se:HBS_algorithm} provides an effective method for creating a sparse representation of a matrix. However, this process yields a set of matrices, whereas a block encoding requires a single matrix. To address this, we use an extended sparsification technique to combine these matrices into a single but larger sparse matrix.

The HBS decomposition yields factors $L, D,$ and $R$ that are either sparse or much smaller than the original matrix. This property holds at each level of the hierarchy, allowing us to approximate the original matrix $A$ with controlled accuracy. We begin by considering a single level of compression. In this case, the matrix $A$ is approximated as $A \approx D+LSR$. The linear system $Ax=b$ can then be solved by replacing $A$ with this approximation:
\[(D+LSR)x=b.\]
To solve this system, we introduce the auxiliary variables $z=Rx$ and $y=Sz$, which allows us to write the system in the form
\begin{equation*}
    \begin{pmatrix}
        D & L & \\
        R & & -I\\
         & -I& S\\
    \end{pmatrix}
    \begin{pmatrix}
        x\\
        y\\
        z
    \end{pmatrix}
    =\begin{pmatrix}
        b\\
        0\\
        0
    \end{pmatrix}.
\end{equation*}

Applying this procedure recursively to the telescoping representation in Equation \eqref{eq:recursive} yields the following extended sparse system:
\begin{widetext}
\begin{equation}\label{eq:extended}
    \begin{pmatrix}
        D_1&L_1&&&&&\\
        R_1&&-I&&&&\\
        &-I&D_2&L_2&&\\
        &&R_2&\ddots&\ddots&&\\
        &&&\ddots&D_\lambda&L_\lambda&\\
        &&&&R_\lambda&&-I\\
        &&&&&-I&D_{\lambda+1}\\
    \end{pmatrix}
    \begin{pmatrix}
        x\\
        y_1\\
        z_1\\
        \vdots\\
        \vdots\\
        y_\lambda\\
        z_\lambda\\
    \end{pmatrix}=
    \begin{pmatrix}
        b\\
        0\\
        0\\
        \vdots\\
        \vdots\\
        0\\
        0\\
    \end{pmatrix}.
\end{equation}
\end{widetext}
We denote the block matrix in this system by $A_{\text{sp}}$ and refer to it as the extended sparsification of $A$. This structure has been employed in \cite{hoFastDirectSolver2012} for the construction of a fast direct solver for boundary integral equations.

To be used in a quantum algorithm, we must analyze the properties of $A_{\text{sp}}$. As shown in Theorem \ref{thm:hss_sparsification_properties}, the total number of nonzero elements in $A_{\text{sp}}$ scales as $O(N)$. We now consider its row and column sparsity. The nonzero entries in any given row of $A_{\text{sp}}$ come from at most three types of block matrices: a diagonal block $D_l$, the interpolation matrices $L_l$ and $R_{l-1}$, and an identity matrix $I$. By construction, the interpolation matrices at level $l$ are sparse, with at most $k_l$ or $k_l+1$ nonzeros per row or column. The diagonal blocks $D_l$ contain near-neighbor interactions. At the finest level ($l=1$), a block interacts with its $3^d-1$ neighbors, so each row in $D_1$ has at most $3^d n_1$ nonzeros. For coarser levels ($l>1$), a block interacts with its $(6^d-3^d)$ neighbors, resulting in at most $(6^d-3^d)n_l$ nonzeros per row in $D_l$. Combining these bounds, the maximum number of nonzeros in any row of $A_{\text{sp}}$ is
\begin{widetext}
\begin{equation}
	s_r=\max\{3^dn_1+k_1,n_1-k_1+1,(6^d-3^d)k_1+k_2+1,n_2-k_2+1,\cdots,n_\lambda-k_\lambda+1,(6^d-3^d)k_{\lambda}+1\}.
\end{equation}
\end{widetext}
Since $n_l$ and $k_l$ are constants for all levels, the row sparsity $s_r$ is bounded by a constant. A similar argument holds for the column sparsity $s_c$.

Finally, we consider the dimension of the matrix $A_{\text{sp}}$. The total number of rows, $N_{sp}$, is the sum of the dimensions of all the subblocks, which also scales linearly with $N$:
\begin{equation}
    N_{\text{sp}}=\sum\limits_{l=1}^{\lambda}2p_lk_l+p_1n_1\leq 2r(2^{d(\lambda+1)}-2^d)+N\sim O(N).
\end{equation}

With these properties established, we can determine the cost of creating a block encoding for $A_{\text{sp}}$. Let $c_{\text{sp}}$ be the upper bound on the magnitude of elements in the original matrix $A$, and let $f$ be the parameter defined in Theorem \ref{apth:RRQR} for the strong RRQR algorithm~\cite{guEfficientAlgorithmsComputing1996}. Based on the conditions of Theorem \ref{apth:block_encoding}, we arrive at the following theorem:
\begin{theorem}\label{thm:sparsification_block}
    Let $A$ be a matrix represented hierarchically to depth $\lambda$ with the extended sparsification $A_{\rm{sp}}=(a_{ij})_{N_{\rm{sp}}\times N_{\rm{sp}}}$ that is $s_r$-row-sparse and $s_c$-column-sparse. Then, one can implement an $(c_{\rm{sp}}\sqrt{s_rs_c},O(n),\varepsilon)$ block encoding of $A_{\rm{sp}}$ with a single query to $\mathcal{O}_r$ and $\mathcal{O}_c$, two queries to $\mathcal{O}_A$, and $O(n+{\rm polylog}(\frac{c_{\rm sp}\sqrt{s_r s_c}}{\varepsilon}))$ additional one and two qubit gates, using $O(b,{\rm polylog}(\frac{c_{\rm sp}\sqrt{s_r s_c}}{\varepsilon}))$ ancilla qubits. The oracles $O_r,O_c,O_A$ provide access to $O(N)$ nonzero matrix entries, where $N$ is the dimension of $A$.
\end{theorem}

\begin{proof}
	By Theorem~\ref{thm:hss_sparsification_properties}, the extended matrix $A_{\rm sp}$ stores only $O(N)$ nonzeros and has bounded row/column sparsities $s_r,s_c$, with entry magnitudes bounded by $c_{\rm sp}$. 
	Applying Lemma~\ref{apth:block_encoding} to $A_{\rm sp}$ yields an $(c_{\rm{sp}}\sqrt{s_rs_c},\,\cdot,\,\varepsilon)$-block-encoding with the stated oracle/query costs.
\end{proof}

While this sparsification method allows us to handle the dense matrix with algorithms designed for sparse systems, it also introduces three challenges. First, the solution vector of the extended system contains not only the desired solution $x$ but also the auxiliary variables $y_l$ and $z_l$. These additional components reduce the probability of measuring $x$, which increases the overall runtime of the algorithm. Second, the complexity of quantum linear system solvers is sensitive to the condition number of the matrix. The relationship between the condition number of the extended matrix $A_{\text{sp}}$ and the original matrix $A$ is not obvious, and the transformation could potentially make it much larger. Finally, the HBS representation is an approximation of the original dense matrix, and the impact of this initial error on the accuracy of the final quantum solution needs to be analyzed. Theorem \ref{thm:sparsification_regularization_error} provides quantitative bounds for each of these issues.

\begin{theorem}\label{thm:sparsification_regularization_error}
Let $A \in \mathbb{C}^{N \times N}$ be a matrix approximated by a hierarchical matrix $A_\varepsilon$ with relative error $\|A - A_\varepsilon\|/\|A\| \le \varepsilon$. Let $A_{\rm{sp}}$ be the extended sparse matrix derived from $A_\varepsilon$ with row sparsity $s_r$, column sparsity $s_c$, and entries bounded by $|(A_{\rm{sp}})_{ij}| \le c_{\rm{sp}}$. The following bounds hold: 
\begin{enumerate}[label=(\arabic*)]
    \item By applying a scaling with parameter $0 < t < 1$ in the hierarchical decomposition, the subnormalization factor of the resulting block encoding for the new sparse matrix $A'_{\rm{sp}}$, which is obtained by replacing $L_l$ with $t^{-1}L_l$ and $R_l$ with $tR_l$ in $A_{\rm{sp}}$, is bounded by $\alpha(A'_{\rm{sp}}) \le t^{-1}c_{\rm{sp}}\sqrt{s_rs_c}$. The squared norm of the auxiliary variables $(y', z')$ in the scaled solution $x' = (x^T, (y')^T, (z')^T)^T$ is constrained by the following inequality:
    \[\sum_{l=1}^\lambda (\|y'_l\|^2 + \|z'_l\|^2) \le t^2 \sum_{l=1}^\lambda (\|y_l\|^2 + \|z_l\|^2).\]
    \item For any regularization parameter $\alpha > 0$, the condition number of the Tikhonov regularized normal matrix $A_{\rm{sp}}^\dagger A_{\rm{sp}} + \alpha I$ is bounded by:
\begin{equation}
\kappa(A_{\rm{sp}}^\dagger A_{\rm{sp}} + \alpha I) \le \frac{c_{\rm{sp}}^2 s_r s_c + \alpha}{\alpha}.
\end{equation}
    \item Let $x = A^{-1}b$ and $x_\varepsilon = A_\varepsilon^{-1}b$. If the approximation is sufficiently accurate such that $\varepsilon\kappa(A) < 1$, the relative error in the solution is bounded by:
\begin{equation}
\frac{\|x - x_\varepsilon\|}{\|x\|} \le \frac{\varepsilon\kappa(A)}{1 - \varepsilon\kappa(A)}.
\end{equation}
\end{enumerate}
\end{theorem}

Theorem \ref{thm:sparsification_regularization_error} shows that the costs associated with this sparsification method are controllable. The condition number can be bounded with regularization, and the approximation error is tied to the properties of the original system. These benefits come at a quantifiable cost in the success probability of the measurement.

We now justify the three claims in Theorem~\ref{thm:sparsification_regularization_error}. 
For readability, the argument is organized into three corresponding parts in Section~\ref{apse:2.1}.

\subsection{Analysis for the Sparsification Based Quantum Algorithm}\label{apse:2.1}

The extended sparsification approach transforms the dense matrix into a larger but sparse linear system. While this structure is amenable to known quantum solvers, the transformation itself introduces practical issues. Here, we analyze three main aspects: the impact of auxiliary variables on the measurement success probability, the condition number of the new sparse matrix, and the error resulting from the initial hierarchical approximation.

\subsubsection{Suppressing Auxiliary Variable Norms}

Solving the extended linear system gives the solution vector $x'=(x^T, y_1^T, z_1^T, \ldots, y_\lambda^T, z_\lambda^T)^T$. When this vector is represented as a quantum state, the desired solution $x$ is only one component of the full state $x'$. This means that a measurement on the final state $x'$ yields the part corresponding to $x$ with a success probability of only $\frac{\Vert x\Vert_2^2}{\Vert x'\Vert_2^2}$. The norms of the auxiliary variables $y_l$ and $z_l$ decrease this success probability, which in turn increases the number of repetitions needed and thus the overall runtime. We propose two methods to address this problem.

The first method, which is stated in Theorem \ref{thm:sparsification_regularization_error} is to rescale the linear system. We can modify the HBS factorization by introducing a scaling parameter $t$:
\begin{widetext}
	\begin{equation}\label{eq:reformulate}
		A\approx D_1+\left(\frac{1}{t}L_1\right)\left[D_2+\left(\frac{1}{t}L_2\right)\left(\cdots D_\lambda+\left(\frac{1}{t}L_\lambda\right) D_{\lambda+1} (tR_\lambda)\cdots\right)(tR_2)\right](tR_1).
	\end{equation}
\end{widetext}
This is an exact reformulation and does not introduce additional approximation error. It modifies the extended linear system so that the new auxiliary variables become $y'_l = t^l y_l$ and $z'_l = t^l z_l$. By choosing a parameter $0< t < 1$, the norms of the auxiliary components are suppressed exponentially with the level $l$. This increases the success probability of measuring the state for $x$. Using this modification, we have 
\begin{widetext}
	\begin{equation}
		p_{\text{succ}}=\frac{\Vert x\Vert ^2_2}{\Vert x'\Vert_2^2}=\frac{\Vert x\Vert^2}{\Vert x\Vert^2 +\sum_l^\lambda (\Vert y'_l\Vert^2+\Vert z'_l\Vert^2)}\geq \frac{\Vert x\Vert ^2}{\Vert x\Vert^2 +t^2\sum_l^\lambda (\Vert y_l\Vert^2+\Vert z_l\Vert^2)}.
	\end{equation}
\end{widetext}
As $t$ is reduced, the failure probability decreases.  This establishes item (1) in Theorem \ref{thm:sparsification_regularization_error}.

The second method is to apply a post-processing operator that removes the auxiliary components from the solution vector $x'$. This is possible because the relationships between $x$, $y_l$, and $z_l$ are defined by Equation \eqref{eq:extended}. We can construct an invertible matrix $A'$ that, when applied to $x'$, produces a vector containing only the $x$ component:
\begin{equation}
	A'\cdot x'=\begin{pmatrix}
		I&&&&&&\\
		R_1&&-I&&&&\\
		&-I&D_2&L_2&&\\
		&&R_2&\ddots&\ddots&&\\
		&&&\ddots&D_\lambda&L_\lambda&\\
		&&&&R_\lambda&&-I\\
		&&&&&-I&D_{\lambda+1}\\
	\end{pmatrix}
	\begin{pmatrix}
		x\\
		y_1\\
		z_1\\
		\vdots\\
		\vdots\\
		y_\lambda\\
		z_\lambda\\
	\end{pmatrix}=
	\begin{pmatrix}
		x\\
		0\\
		0\\
		\vdots\\
		\vdots\\
		0\\
		0\\
	\end{pmatrix}.
\end{equation}
This matrix $A'$ is block sparse with a structure similar to $A_{\text{sp}}$, and therefore it can be block-encoded using Theorem \ref{thm:sparsification_block}. This approach involves first preparing the state $\ket{x'}$ and then applying the block encoding of $A'$ to transform it into the state corresponding to $x$. The cost of this method comes from constructing and applying this second block encoding, which depends on its subnormalization factor $\alpha(A')\approx\alpha(A)$.

\subsubsection{Control of Condition Number}

As noted, a potential issue with the extended sparsification is that the matrix $A_{\text{sp}}$ may be ill-conditioned. We discuss two methods to address this.

One approach is to use a preconditioner. If a sparse preconditioner $M$ can be found, we can instead solve the equivalent system:
\[MAx=Mb.\]
A well-chosen preconditioner ensures that the condition number of $MA$ is smaller than that of $A_{\text{sp}}$. For example, incomplete LU factorization (ILU) with thresholding constructs a sparse approximate LU decomposition of a matrix. Its parameters can be tuned to balance the reduction in the condition number against the element number of $M$. Further information can be found in \cite{doi:10.1137/1.9781421407944,wathenPreconditioning2015}.

An alternative approach is to use Tikhonov regularization to directly control the condition number, as shown in Theorem \ref{thm:sparsification_regularization_error}. Before introducing this method, we firstly state Gershgorin circle theorem \cite{doi:10.1137/1.9781421407944}, which provides a useful bound on the spectral radius of matrix:
\begin{lemma}\label{co:Gerschgorin}
	For any $A\in\mathbb{C}^{n\times n}$, its spectrum is bounded by:
	\[{\rm spectrum}(A)\subseteq\{z\in\mathbb{C}:\vert z\vert\leq\max\limits_i\sum\limits_{j=1}^n\vert a_{ij}\vert\}.\]
\end{lemma}
We can apply this result directly to $A_{\text{sp}}$. From Theorem \ref{thm:hss_sparsification_properties}, the entries of $A_{\text{sp}}$ can be bounded by some constant $c_{\rm sp}>0$. Since each row of $A_{\text{sp}}$ has a constant sparsity $s_r$, Lemma \ref{co:Gerschgorin} implies that any eigenvalue $z \in {\rm spectrum}(A_{\text{sp}})$ satisfies $\vert z\vert\leq c_{\rm sp}s_r$. Therefore, the largest eigenvalue of $A_{\text{sp}}$ is bounded.

However, Lemma \ref{co:Gerschgorin} does not provide a lower bound on the smallest eigenvalue. To manage this problem, we employ Tikhonov regularization \cite{englRegularizationInverseProblems1996}. Instead of solving $A_{\text{sp}}x=b$ directly, we solve a modified system. For a small parameter $\alpha > 0$, we solve $(A^\dagger_{\text{sp}} A_{\text{sp}}+\alpha I)x'=A^\dagger_{\text{sp}} b$, yielding the regularized solution $x'=(A_{\text{sp}}^\dagger A_{\text{sp}}+\alpha I)^{-1}A_{\text{sp}}^\dagger b$. Let $R_\alpha=(A^\dagger_{\text{sp}} A_{\text{sp}}+\alpha I)^{-1}A^\dagger_{\text{sp}} $. The following theorem guarantees that this regularized solution converges to the true solution.
\begin{theorem}[\cite{englRegularizationInverseProblems1996}]\label{apth:Tikhonov}
	Let $A:X\to Y$ be an injective compact linear operator with dense range in $Y$. Let $f\in A(X)$ and $f^\delta\in Y$ satisfy
	\[\Vert f^\delta -f\Vert\leq \delta <\Vert f^\delta \Vert\]
	with $\delta>0$. Then there exists a unique parameter $\alpha=\alpha(\delta)$ such that
	\[\Vert AR_{\alpha(\delta)}f^\delta-f^\delta\Vert =\delta\]
	is satisfied and
	\[R_{\alpha(\delta)}f^\delta\to A^{-1}f,\quad \delta\to 0.\]
\end{theorem}
Theorem \ref{apth:Tikhonov} shows that the regularized solution converges to the true solution as the error $\delta$ in the right-hand side goes to zero. Therefore with an appropriate choice of $\alpha$, we can obtain a highly accurate solution.

We now analyze the condition number of the regularized matrix $(A^\dagger_{\text{sp}} A_{\text{sp}}+\alpha I)$. From Lemma \ref{co:Gerschgorin}, we know that $\Vert A_{\text{sp}}\Vert_2\leq c_{\text{sp}}s_r$, so 
\[\Vert A_{\text{sp}}^\dagger A_{\text{sp}}+\alpha I\Vert_2\leq\Vert A^\dagger_{\text{sp}} \Vert_2\Vert A_{\text{sp}}\Vert_2+\alpha\Vert I\Vert_2\leq c_{\text{sp}}^2s_r^2+\alpha.\]
For the lower bound, for any unit vector $x$, we have $x^\dagger (A_{\text{sp}}^\dagger A_{\text{sp}}+\alpha I)x\geq \alpha $, which implies that the minimum eigenvalue of $(A^\dagger_{\text{sp}} A_{\text{sp}}+\alpha I)$ is bounded below by $\alpha$. Consequently, the condition number is bounded by: 

\[\kappa(A_{\text{sp}}^\dagger A_{\text{sp}}+\alpha I)\leq\frac{c_{\text{sp}}^2s_r^2+\alpha}{\alpha}.\]
Since $(A_{\text{sp}}^\dagger A_{\text{sp}}+\alpha I)$ is the matrix to be inverted when solving the system of equations $(A_{\text{sp}}^\dagger A_{\text{sp}}+\alpha I)x'=A_{\text{sp}}^\dagger b$ in the Tikhonov regularization step, this bound on its condition number is sufficient for solving the system. This proves item (2) in Theorem \ref{thm:sparsification_regularization_error}.

\subsubsection{Error from Hierarchical Matrix Approximation}
Since the extended system is built from an HBS approximation of the original matrix $A$, we also analyze the error introduced by this approximation. The error in the final solution can be bounded using a standard result from numerical linear algebra \cite{doi:10.1137/1.9781421407944}.

\begin{theorem}\label{apth:error}
	Let $A$ be the original matrix and $A_\varepsilon$ its compressed representation. Assume that ${\Vert A-A_\varepsilon\Vert}/{\Vert A\Vert}\leq \varepsilon$ and the condition $\varepsilon\kappa(A)<1$ holds, which can be ensured by the compression algorithm. Let $x$ and $b$ be vectors such that $Ax=b$. Define $b_\varepsilon=A_\varepsilon x$ and $x_\varepsilon=A_\varepsilon^{-1}b$. Then the following error bounds hold:
	\[\frac{\Vert b-b_\varepsilon\Vert}{\Vert b\Vert}\leq\varepsilon\kappa(A),\]
	\[\frac{\Vert x-x_\varepsilon\Vert}{\Vert x\Vert}\leq\frac{\varepsilon\kappa(A)}{1-\varepsilon\kappa(A)},\]
	where $\kappa(A)$ is the condition number of $A$.
\end{theorem}
Theorem \ref{apth:error} shows that if the matrix $A$ is well-conditioned, the relative error in the solution is of the order $O(\varepsilon)$. In practice, the approximation accuracy $\varepsilon$ is chosen to be sufficiently small to ensure that the condition $\varepsilon\kappa(A) < 1$ holds. Item (3) in Theorem \ref{thm:sparsification_regularization_error} is exactly Theorem~\ref{apth:error}.

\section{A Direct Block Encoding Algorithm for Hierarchical Matrices}\label{se:direct_be}

An alternative to the extended sparsification method is to construct a block encoding of the matrix directly from its HBS approximation. Instead of forming an extended sparse system to solve a linear equation, this approach embeds the matrix itself into a unitary operator, as defined in Definition \ref{def:block_encoding}. This method requires oracles that access only $O(N)$ nonzero elements, in contrast to the $O(N^2)$ elements required by the approach in \cite{nguyenBlockencodingDenseFullrank2022}, and it works for a broader class of kernel functions. The construction follows the hierarchical structure of the matrix shown in Equation \eqref{eq:recursive}, proceeding from the lowest level upwards to build the encoding for the full matrix.

The construction relies on two main quantum primitives: the linear combination of unitaries (LCU) and quantum matrix multiplication (QMM). See Figures   \ref{fig:lcu_circuit} and \ref{fig:qmm_circuit} for their quantum circuits, respectively. We assume that block encodings for the component matrices $D_l$, $L_l$, and $R_l$ are given, denoted as $U^D_l$, $U^L_l$, and $U_l^R$, respectively. Each encoding uses $a$ ancilla qubits and has an error of $\varepsilon$. The subnormalization factors for the diagonal blocks $D_l$ are denoted $\alpha_l^{(D)}$, and for the interpolation matrices $L_l, R_l$ they are $\alpha_l^{(L)}$. Based on Lemma \ref{apth:block_encoding} for block encoding sparse matrices, and recalling that the elements of $D_l$ are drawn from $A$, we have
\begin{equation}
\begin{aligned}
    &\alpha_l^{(D)} \le (6^d-3^d)n_lc_{\text{sp}},\ l>1,\\ &\alpha_1^{(D)} \le 3^dn_1c_{\text{sp}}.
\end{aligned}
\end{equation}
For the interpolation matrices, the subnormalization factor is bounded by 
\begin{equation}
    \alpha_l^{(R)}=\alpha_l^{(L)} \le f\sqrt{k_l(n_l-k_l+1)}.
\end{equation}
Since $n_l$ and $k_l$ are bounded by constants, these subnormalization factors are also bounded. To combine these block encoded matrices using LCU, we first need the concept of a state preparation pair.

\begin{definition}[State preparation pair]
    Let $y\in\mathbb{C}^m$ with $\Vert y\Vert_1\leq\beta$. A pair of unitaries $(P_L,P_R)$ acting on $n$ qubits is called a $(\beta,n,\gamma)$-state-preparation-pair for $y$ if $P_L\ket{0}^{\otimes n}=\sum_{j=0}^{2^n-1}c_j\ket{j}$ and $P_R\ket{0}^{\otimes n}=\sum_{j=0}^{2^n-1}d_j\ket{j}$ such that $\sum_{j=0}^{m-1}\vert \beta c_j^*d_j-y_j\vert\leq \gamma$ and $c_j^*d_j=0$ for any $j\in\{m,\ldots,2^n-1\}$.
\end{definition}

The LCU technique allows for the construction of a linear combination of unitary operators, which in turn implements the addition of the matrices they block encode.

\begin{lemma}[LCU \cite{gilyenQuantumSingularValue2019}]\label{apth:LCU}
    Let $A=\sum_{j=0}^{m-1}y_jA_j$ be an $s$-qubit operator where $\Vert y\Vert_1\leq\beta$. Suppose $(P_L, P_R)$ is a $(\beta,n,\gamma)$-state-preparation-pair for $y$, and $W=\sum_{j=0}^{m-1}\ket{j}\bra{j}\otimes U_j+(I-\sum_{j=0}^{m-1}\ket{j}\bra{j})\otimes I_a\otimes I_s$ is an $(n+a+s)$-qubit unitary such that for all $j\in\{0,\ldots,m-1\}$, $U_j$ is an $(\alpha,a,\varepsilon)$-block-encoding of $A_j$. Then the unitary $\widetilde{W}=(P_L^\dagger\otimes I_a\otimes I_s)W(P_R\otimes I_a\otimes I_s)$ is an $(\alpha\beta,a+n,\alpha\gamma+\beta\varepsilon)$-block-encoding of $A$.
\end{lemma}

\begin{figure}[htbp]
\centering
\[
\Qcircuit @C=1.2em @R=1.0em {
    \lstick{\ket{0}^{\otimes n}} & \gate{P_R} & \ctrl{1} & \gate{P_L^\dagger} & \qw \\
    \lstick{\ket{0}^{\otimes a}} & \qw        & \multigate{1}{W} & \qw & \qw \\
    \lstick{\ket{\psi}_s}        & \qw        & \ghost{W} & \qw & \qw
}
\]
\caption{Quantum circuit for LCU.}\label{fig:lcu_circuit}
\end{figure}
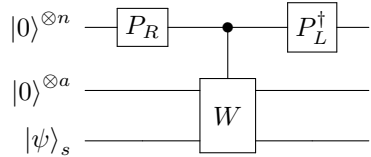

Lemma \ref{apth:LCU} assumes a uniform subnormalization factor $\alpha$ for all operators. This result can be extended to the more common case where we have a set of block encodings $\{U_j\}$, each being an $(\alpha_j, a, \varepsilon_j)$-encoding for its respective matrix $A_j$. We can choose a reference subnormalization factor $\alpha = \max_j \alpha_j$ to rescale the matrices and coefficients. By defining new coefficients $y'_j = y_j \alpha_j / \alpha$, the linear combination for $A$ becomes $A=\sum_{j=0}^{m-1}y'_j(A_j{\alpha}/{\alpha_j})$. The effective subnormalization factor for the block encoding of $A$ is then the weighted sum $\sum_{j=0}^{m-1} |y_j|\alpha_j$. The combined error is bounded by $\alpha\gamma+\sum_{j=0}^{m-1}|y_j|\varepsilon_j$. Thus, the operator $\widetilde{W}$ is an $(\sum_{j=0}^{m-1}\vert y_j\vert\alpha_j, a+n, \alpha\gamma+\sum_{j=0}^{m-1}|y_j|\varepsilon_j)$-block-encoding of $A$.

Beyond combining matrices linearly, it is also essential to handle their products. The following lemma, known as QMM, provides a straightforward way to construct the block encoding of a product $AB$ by composing the block encodings of $A$ and $B$.
\begin{lemma}[QMM \cite{gilyenQuantumSingularValue2019}]\label{apth:QMM}
    If $U$ is an $(\alpha,a,\delta)$-block-encoding of an $s$-qubit operator $A$, and $V$ is an $(\beta,b,\varepsilon)$-block-encoding of an $s$-qubit operator $B$ then $(I_b\otimes U)(I_a\otimes V)$ is an $(\alpha\beta,a+b,\alpha\varepsilon+\beta\delta)$-block-encoding of $AB$. Here $I_a(I_b)$ acts on the ancilla qubits of $U(V)$.
\end{lemma}

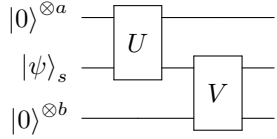
\begin{figure}[htbp]
\centering
\[
\Qcircuit @C=1.2em @R=1.0em {
    \lstick{\ket{0}^{\otimes a}} & \multigate{1}{U} &  \qw & \qw \\
    \lstick{\ket{\psi}_s} & \ghost{U}        & \multigate{1}{V} & \qw \\
    \lstick{\ket{0}^{\otimes b}}        & \qw        & \ghost{V} & \qw
}
\]
\caption{Quantum circuit for QMM.}\label{fig:qmm_circuit}
\end{figure}

We now describe the direct construction of a block encoding for $A$ from its HBS factors. To construct the block encoding at level $l$, we use the oracles $U^D_{l+1}$, $U^L_l$, and $U^R_l$ and combine them using QMM and LCU to implement the recursive formula in Equation \eqref{eq:recursive}. The full construction for $A$ involves $2\lambda$ multiplications and $\lambda+1$ additions of block-encoded matrices, requiring a total of $2a\lambda+\lambda+1=O(\log(N)a)$ ancilla qubits. The overall properties are summarized in the following theorem.
\begin{theorem}\label{thm:recursive_be_analysis}
Let $A$ be a matrix represented hierarchically to depth $\lambda$. Let $\{D_l\},\{L_l\},\{R_l\}$ be the HBS decomposition of $A$ and $U_l^D$ ($U_l^L,U_l^R$) be a $(\alpha_l^{(D)},a,\varepsilon)$ ($(\alpha_l^{(L)},a,\varepsilon),(\alpha_l^{(R)},a,\varepsilon)$) block encoding of $D_l$ ($L_l,R_l$). Assume the HBS representation is sparse, such that the total number of nonzero entries across all matrices $\{D_l\}, \{L_l\}, \{R_l\}$ is $O(N)$. Further, we assume that $\alpha_l^{(L)}=\alpha_l^{(R)}$. Then with QMM and LCU, a $(\alpha_A,O(\log{N})a,\varepsilon_A)$ block encoding of $A$ is obtained, where
\begin{equation}\label{eq:subnormal_direct}
\alpha_A=\sum\limits_{l=1}^{\lambda+1}\alpha_l^{(D)}\prod\limits_{m=1}^{l-1} (\alpha_m^{(L)})^2,
\end{equation}
and
\begin{equation}\label{eq:error_direct}
\varepsilon_A =\left[\sum\limits_{l=1}^{\lambda}\left(2\alpha^{(L)}_l\sum_{j=l}^{\lambda+1}\alpha_j^{(D)}\prod_{m=1}^{j-1}(\alpha_m^{(L)})^2+1\right)\prod\limits_{m=1}^{l-1}(\alpha_m^{(L)})^2+\prod_{m=1}^\lambda(\alpha_m^{(L)})^2\right]\varepsilon.
\end{equation}
\end{theorem}
\begin{proof}
	We follow the telescoping reconstruction in \eqref{eq:recursive} from level $\lambda+1$ upward. 
	At each level, QMM accounts for the product $L_lA_{l+1}R_l$ and multiplies the corresponding subnormalization factors, while LCU adds the diagonal contribution $D_l$ and adds subnormalizations linearly. 
	This yields the recursion for $\alpha_l$ and the claimed closed form \eqref{eq:subnormal_direct}. 
	The precision bound \eqref{eq:error_direct} follows by repeatedly applying Lemma~\ref{apth:QMM} and Lemma~\ref{apth:LCU} and collecting the induced additive errors at each multiplication/addition step.
\end{proof}

For completeness, we illustrate the recursion and the resulting upper bounds. First, we establish a bound for the subnormalization factor. The matrix is reconstructed recursively from the innermost level, following the relation from Equation \eqref{eq:recursive}:
\begin{equation}
	A_l=L_lA_{l+1}R_l+D_l,\  l=1,\ldots,\lambda,\quad A_{\lambda+1}=D_{\lambda+1}.
\end{equation}
Let $\alpha_l$ be the subnormalization factor for the block encoding of $A_l$, and $\{\alpha^{(L)}_l,\alpha^{(D)}_l\}$ be as defined in Section \ref{se:sparsification_quantum}. The product term $L_lA_{l+1}R_l$ is implemented using QMM, which multiplies the subnormalization factors of the components. The addition of the $D_l$ term is done using LCU. This leads to the relation:
\begin{equation}
	\alpha_l=\alpha_{l+1}(\alpha_l^{(L)})^2+\alpha_l^{(D)},\ l=1,\ldots,\lambda,\quad \alpha_{\lambda+1}=\alpha_{\lambda+1}^{(D)}.
\end{equation}
Solving this relation iteratively gives the result \eqref{eq:subnormal_direct}.

A large subnormalization factor can reduce the efficiency of quantum algorithms. To address this, we can use a technique to amplify the singular values of the encoded matrix, which effectively reduces the subnormalization factor. The following lemma, as a corollary from \cite{gilyenQuantumSingularValue2019}, provides the basis for this method.

\begin{lemma}[Uniform singular value amplification \cite{gilyenQuantumSingularValue2019}]\label{apth:singular_value_amplification}
	Let $U$ be a block encoding of a matrix $A$ whose singular values $\zeta_i$ lie in the range $[0,(1-\delta)/\gamma]$ for some $\gamma>1$ and $0<\delta<1/2$. Then, for any $0<\varepsilon<1/2$, there exists an efficiently computable polynomial of degree $O(\frac{\gamma}{\delta}\log\frac{\gamma}{\varepsilon})$ that, when applied to $U$ using the QSVT, results in a new block encoding $\tilde{U}$. This $\tilde{U}$ block-encodes a matrix $\tilde{A}$ that is an $\varepsilon$-approximation of $\gamma A$. Specifically, each amplified singular value $\tilde{\zeta}_i$ of $\tilde{A}$ satisfies $\vert \tilde{\zeta}_i/(\gamma\zeta_i)-1\vert\leq \varepsilon$.
\end{lemma}

While Lemma \ref{apth:singular_value_amplification} can reduce the subnormalization factor by a factor of $\gamma$, this amplification is not without cost. The complexity of the new block encoding, measured in the number of queries to the original encoding $U$, scales with the amplification factor $\gamma$. This trade off must be considered. We can apply this technique selectively, for instance only to the block encodings $U^R_l$ and $U^L_l$, to limit the overhead to a constant factor. However, the $2$-norm of a matrix is a lower bound for the subnormalization factor of any of its block encodings. For the interpolation matrices used here, this norm is at least $1$. Consequently, the subnormalization factor of the composite block encoding still increases at each step of the construction.

Finally, we analyze the precision error introduced by this recursive block encoding procedure. Following Theorem \ref{apth:QMM}, we denote the error at level $l$ as $\varepsilon_l$ and compute the total error recursively. This gives the following bound on the error:
\begin{widetext}
	\begin{equation}
		\varepsilon_l=(\alpha_l^{(L)}(\alpha_{l+1}\varepsilon+\alpha_l^{(L)}\varepsilon_{l+1})+\alpha_{l+1}\alpha_l^{(L)}\varepsilon)+\varepsilon+\max((\alpha^{(L)}_l)^2\alpha_{l+1},\alpha_l^D)\gamma,\quad \varepsilon_{l+1}=\varepsilon,
	\end{equation}
\end{widetext}
where $\gamma$ represents the error from the state preparation. Due to the multiple matrix multiplications, the error depends on the subnormalization factors and can accumulate quickly. Setting $\gamma=0$, we obtain the expression:
\begin{widetext}
	\begin{equation}
		\begin{aligned}
			\varepsilon_1&=\sum\limits_{l=1}^{\lambda}(2\alpha^{(L)}_l\alpha_{l+1}+1)\prod\limits_{m=1}^{l-1}(\alpha_m^{(L)})^2\varepsilon+\prod_{m=1}^\lambda(\alpha_m^{(L)})^2\varepsilon\\
			&=\left[\sum\limits_{l=1}^{\lambda}\left(2\alpha^{(L)}_l\sum_{j=l}^{\lambda+1}\alpha_j^{(D)}\prod_{m=1}^{j-1}(\alpha_m^{(L)})^2+1\right)\prod\limits_{m=1}^{l-1}(\alpha_m^{(L)})^2+\prod_{m=1}^\lambda(\alpha_m^{(L)})^2\right]\varepsilon.
		\end{aligned}
	\end{equation}
\end{widetext}
This error bound is polynomial in $N$. As discussed previously, this represents a theoretical upper bound, and the actual error in practice is often much smaller.

We next comment on two implementation details that are implicit in the above construction: dimensional padding across levels and the use of controlled unitaries in LCU. A practical challenge in this process is that the component matrices at different levels have different dimensions. For example, $U^D_{l+1}$ block encodes a matrix smaller than the full system. To address this dimensional mismatch, we implicitly pads the smaller matrices with zeros without introducing any additional computational complexity. The block encodings take the following forms:
\begin{equation*}
    U_l^L=\begin{pmatrix}
        L^1_l&0&*\\
        L^2_l&0&*\\
        *&*&*
    \end{pmatrix},
    U_l^R=\begin{pmatrix}
        U^1_l&U_l^2&*\\
        0&0&*\\
        *&*&*
    \end{pmatrix},
    U_{l+1}^D=\begin{pmatrix}
        D_{l+1}&C_1&*\\
        C_2&C_3&*\\
        *&*&*
    \end{pmatrix}.
\end{equation*}
Here, the blocks denoted by $*$ and $C_i$ are arbitrary in the sense that their contents are chosen to ensure the entire matrix is unitary. The QMM rule ensures the product correctly places the desired matrix in the top-left block:
\[\begin{pmatrix}
    L^1_lD_{l+1}U_l^1&L^1_lD_{l+1}U_l^2&*\\
    L^2_lD_{l+1}U_l^1&L^2_lD_{l+1}U_l^2&*\\
    *&*&*
\end{pmatrix}.\]
These arbitrary matrices do not affect the target block, and this step does not introduce additional error. Another implementation issue is that the LCU procedure requires controlled unitaries, not the unitaries themselves. These are typically constructed by decomposing the unitary into a sequence of elementary gates and then making each gate controlled. Due to the nested structure of our algorithm, a unitary at level $l$ must be controlled by $l$ qubits, which requires the use of multi-controlled gates. The implementation of multi-controlled gates is discussed in \cite{zindorfEfficientImplementationMulticontrolled2025}. With these controlled unitaries prepared, our algorithm can be implemented accordingly.

Combining the levelwise bounds on $\alpha_l^{(D)}$ and $\alpha_l^{(L)}$ with \eqref{eq:subnormal_direct} yields the explicit estimate
\begin{equation}
\alpha_A\leq2^d (6^d-3^d)rc_{\rm sp}\frac{(r(2^dr-r+1)f^2)^{\lambda+1}-r(2^dr-r+1)f^2}{r(2^dr-r+1)f^2-1}+3^dn_1c_{\rm sp},
\end{equation}
where $c_{\rm sp}$ is the maximum magnitude of the entries in $A$. The term $(r(2^dr-r+1)f^2)^\lambda$ indicates that this upper bound is polynomial in the matrix size $N$. However, a tighter bound and a smaller subnormalization factor can be achieved in specific cases. To compare our method with previous work \cite{nguyenBlockencodingDenseFullrank2022}, we consider a matrix generated by the kernel function $\vert x-y\vert ^{-p}$ for uniformly spaced points on a line, where the finest level cell size is $1$. For this case, the maximum magnitude of an entry in $D_l$ is bounded by $2^{-(l-1)p}$. Substituting this bound into Equation \eqref{eq:subnormal_direct} yields:
\begin{equation}\label{eq:subnormofkernelp}
\begin{aligned}
    \alpha_1&=\sum\limits_{l=1}^{\lambda+1}\alpha_l^{(D)}\prod\limits_{m=1}^{l-1} (\alpha_m^{(L)})^2\\
    &\leq2^d(6^d-3^d)r\frac{(\frac{f^2r(2^dr-r+1)}{2^p})^{\lambda+1}-\frac{f^2r(2^dr-r+1)}{2^p}}{\frac{f^2r(2^dr-r+1)}{2^p}-1}+3^dn_1.
\end{aligned}
\end{equation}
When $2^p>f^2r(2^dr-r+1)$, the right-hand side of Equation \eqref{eq:subnormofkernelp} is bounded by a constant. Since $f>1$ and $r>1$, this condition requires $p>1$. The difference with \cite{nguyenBlockencodingDenseFullrank2022} arises because the subnormalization factors accumulate multiplicatively during the QMM steps. Nevertheless, our method offers other benefits. It reduces the number of nonzero elements that the oracles need to handle, which can lead to a more efficient implementation. Additionally, the general framework is applicable to a wider class of kernel functions. These advantages are confirmed by the numerical results in Section \ref{se:numerical}.

\section{Numerical applications}\label{se:numerical}
Many problems in science and engineering can be formulated as a integral equations. The hierarchical matrix structures discussed previously are well-suited for constructing block encodings of dense matrices that arise from the discretization of integral equations. Solving the resulting linear systems is a promising application for quantum algorithms.

The general form of the second kind Fredholm integral equation is \cite{baoSingularitySwappingMethod2024}:
\begin{equation}\label{eq:integralequation}
    \sigma(x)+\int_{\Omega}K(x,y)\sigma(y)dy=f(x).
\end{equation}
Here, $\sigma(x)$ is the unknown function, $f(x)$ is a given function, and $K(x,y)$ is the kernel of the integral operator. Discretizing Equation \eqref{eq:integralequation} on a set of $N$ points $\{x_i\}$ yields a system of linear equations:
\begin{equation}
    \sigma(x_i)+\sum_{j=1,j\ne i}^N w_j K(x_i,x_j)\sigma(x_j)=f(x_i),\quad 1\le i\le N,
\end{equation}
where $w_j$ are quadrature weights. This can be written in the matrix form $A\sigma=f$, where $A$ is a typically dense matrix.

As a test case, we consider the 2D interior Dirichlet problem for the Helmholtz equation with a wavenumber of $\kappa=40$ on a domain with a starfish-shaped boundary. The boundary is discretized into $512$ panels, and integrals are approximated using Gaussian quadrature with $16$ nodes per panel. To construct the hierarchical matrix, we use an HBS compression tolerance of $\varepsilon=10^{-10}$. We first verify the accuracy of this classical compression. As shown in Figure \ref{fig:real}, the solution obtained with the HBS matrix is in close agreement with the one from a direct solver. The observed difference between the two solutions is on the order of $10^{-8}$.

\begin{figure*}[htbp]
\centering
\subfigure[]{\includegraphics[width=0.3\textwidth]{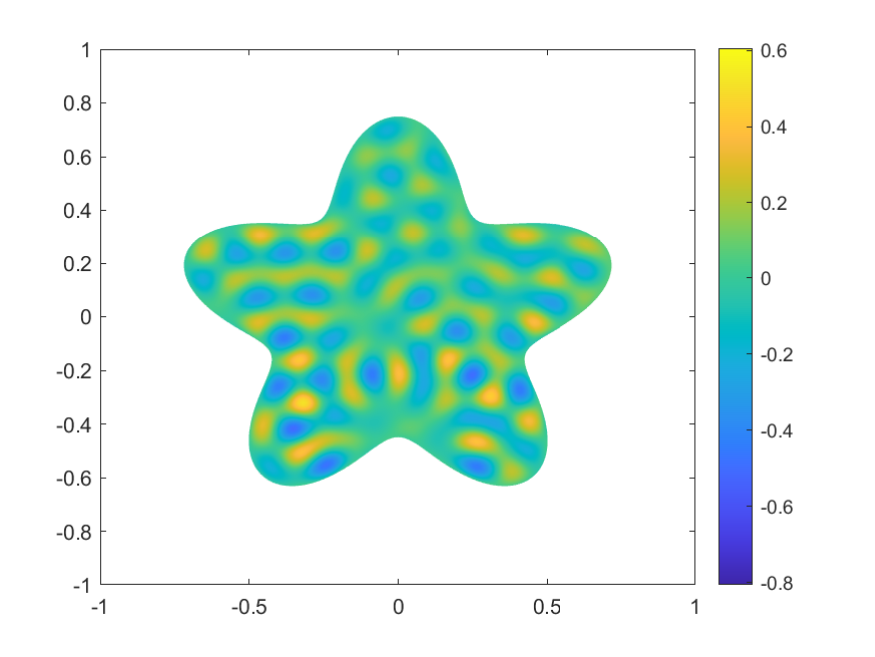}}
\subfigure[]{\includegraphics[width=0.3\textwidth]{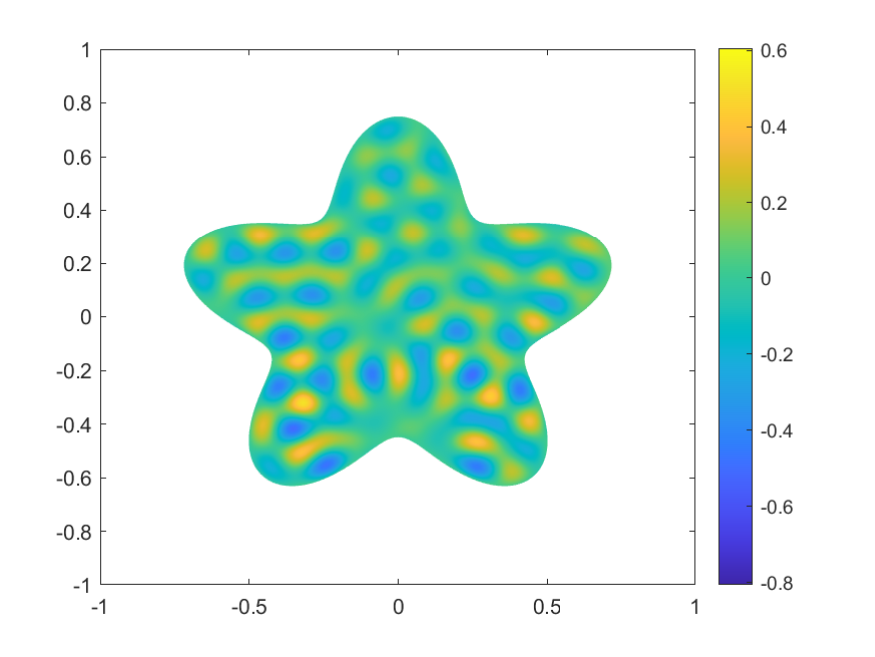}}
\subfigure[]{\includegraphics[width=0.3\textwidth]{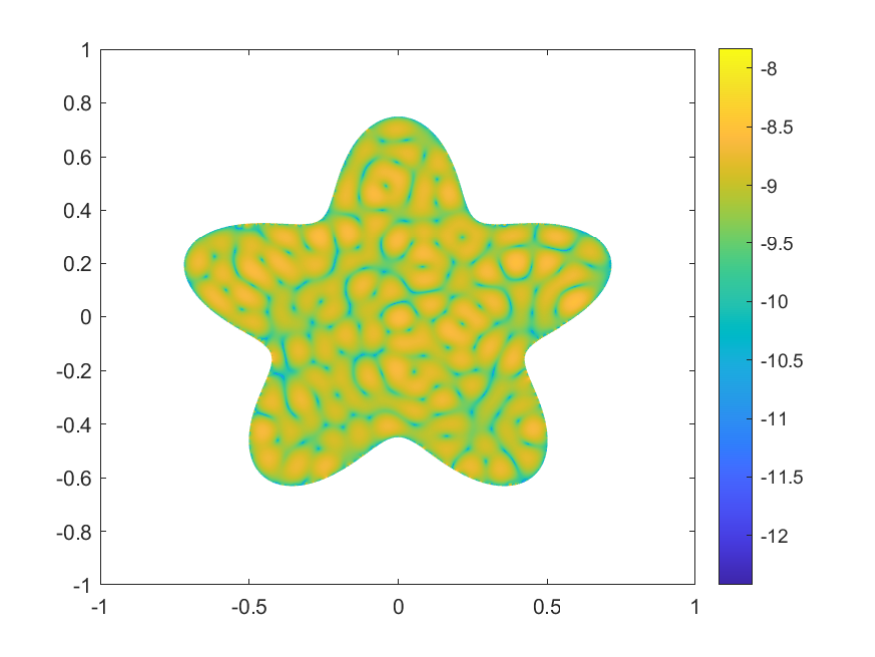}}
\caption{\label{fig:real}{Comparison of the direct solution} (a) and the solution obtained using the HBS-approximated matrix (b) for the Helmholtz equation with Dirichlet boundary conditions. (c) shows the absolute residual in logarithm.}
\end{figure*}

Having confirmed the accuracy of the HBS approximation, we now evaluate metrics relevant to the performance of the quantum algorithm. We use the same starfish-shaped domain but with equispaced points on the boundary, increasing the number of points from $2000$ to $30000$.  Figure \ref{fig:comparison}(a) compares the HBS decomposition time with and without proxy approximation as the problem size $N$ increases. Although the runtime grows superlinearly in both cases, the figure shows that using the proxy technique consistently reduces the computational cost. The condition number of the matrix before and after sparsification is shown in Figure \ref{fig:comparison}(b). As the matrix size increases, the condition number grows, which suggests that a regularization method may be useful for larger problems. The number of nonzero entries after sparsification is shown in Figure \ref{fig:comparison}(c), with the dotted line representing a linear reference. The number of nonzeros grows nearly linearly, indicating that our method effectively compresses the matrix.

\begin{figure*}[htbp]
\centering
\includegraphics[width=0.9\textwidth]{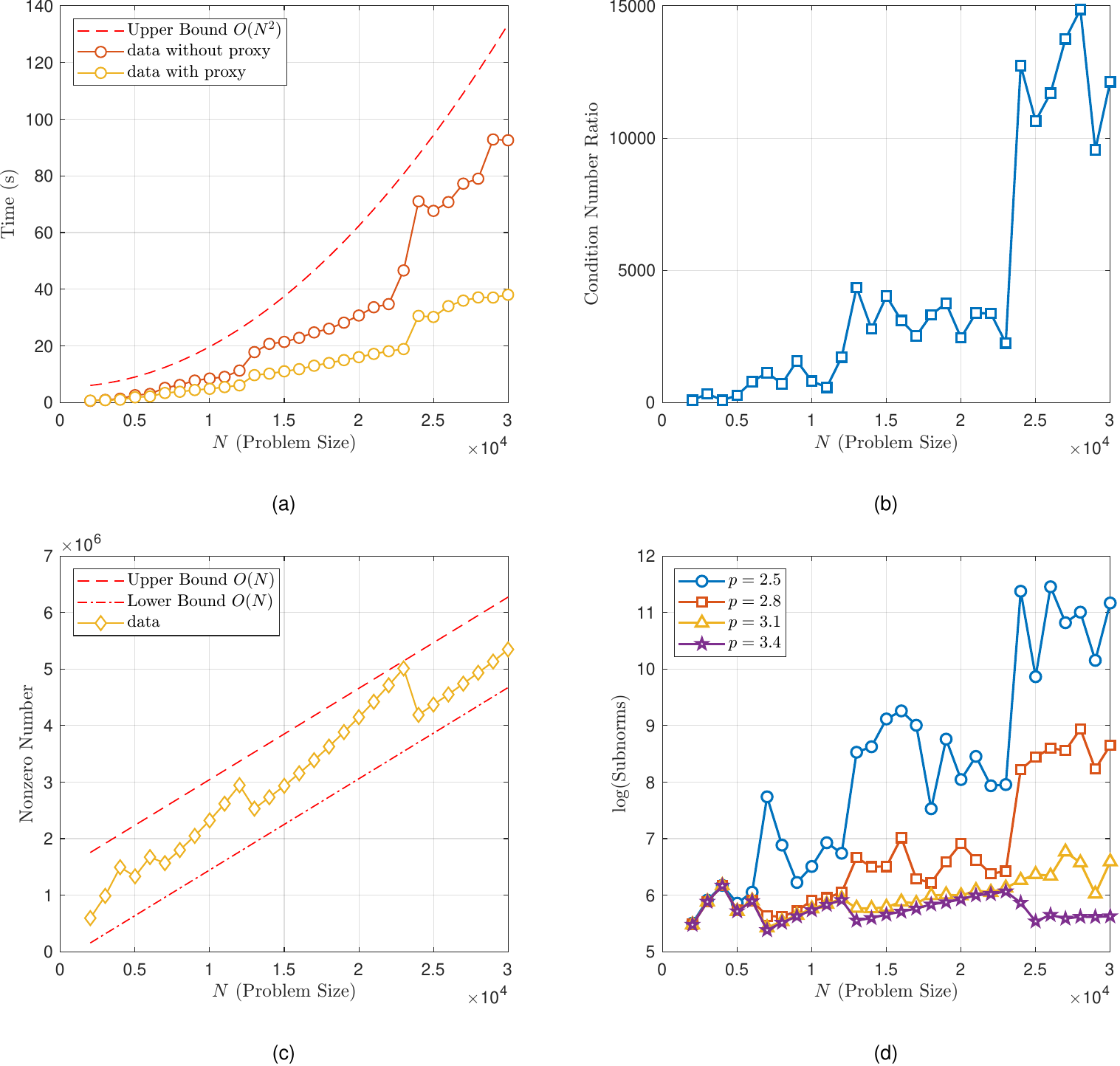}
\caption{\label{fig:comparison} Performance of Algorithm \ref{alg1}. (a) Decomposition runtime as a function of matrix size. (b) Ratio of the condition number of the sparsified matrix to that of the original matrix versus matrix size. (c) Number of nonzero elements as a function of matrix size. (d) Logarithm of the subnormalization factor versus matrix size for the kernel $|x-y|^{-p}$ with different values of $p$.}
\end{figure*}

We also test our algorithm for the kernel $|x-y|^{-p}$ with different values of $p$. As shown in Figure \ref{fig:comparison}(d), we observe that as $p$ increases, the subnormalization factor first grows and then approaches a constant. This result is consistent with the analysis leading to Equation \eqref{eq:subnormofkernelp}.

Finally, to test the scalability of our approach to higher dimensions, we consider the 3D interior Dirichlet problem for the Helmholtz equation on a spherical boundary. We analyze the number of nonzero elements in the matrix generated by our algorithm. As shown in Figure \ref{fig:3d}, this number grows nearly linearly with the matrix size, indicating the method remains effective in 3D.

\begin{figure}[htbp]
\centering
\includegraphics[width=0.5\textwidth]{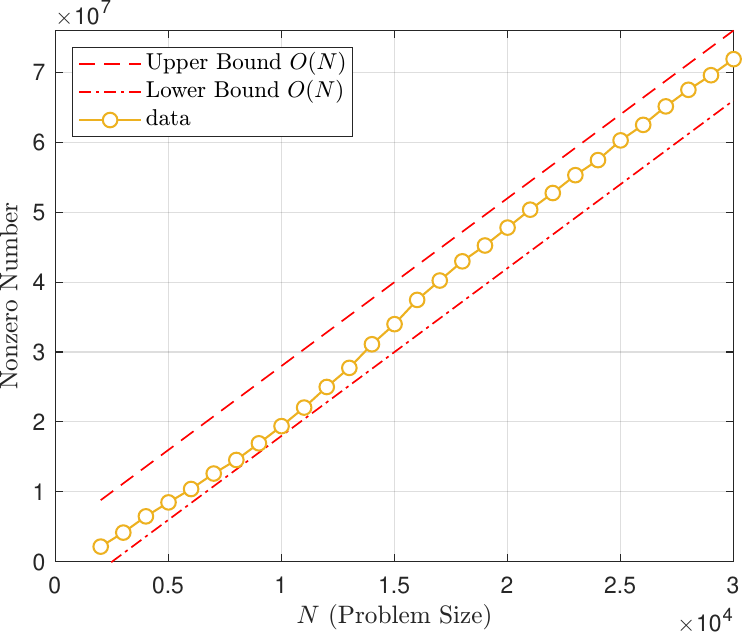}
\caption{\label{fig:3d} Performance of Algorithm \ref{alg1} for the 3D case. The plot shows the number of nonzero elements as a function of matrix size.}
\end{figure}

\section{Conclusions}\label{se:conclusions}
In this work, we addressed the challenge of constructing efficient block encodings for dense matrices with a HBS structure, an important step for applying quantum linear system algorithms to problems in science and engineering. We proposed and analyzed two distinct methods to achieve this. The first approach is based on an extended sparsification technique, which transforms the dense HBS matrix into a larger but sparse system. The second is a direct method that recursively constructs the block encoding by using the hierarchical decomposition of the matrix. Our numerical experiments on integral equations confirm the viability of these approaches.

A primary advantage of our approaches is their ability to reduce the effective number of matrix nonzeros to a total that scales linearly with the problem size. This is crucial for lowering the oracle complexity required by the quantum algorithm. While constructing this sparse representation involves a classical preprocessing cost that also scales as $O(N)$, solving the resulting system on a quantum computer can still provide a huge speedup. This is because classical solvers for such problems can be slow due to large constant factors. Future research could focus on developing preconditioning strategies compatible with the sparsification approach or investigating methods to control the subnormalization factor and error propagation in the direct construction. As block encoding is a fundamental primitive in quantum computing, we expect these methods for HBS matrices will expand the range of practical problems solvable by future quantum algorithms.

\acknowledgements{The work of JL was partially supported by the National Key Research and Development Program of China (2025YFA1016800), NSFC grant No. 12371427 and the “Xiaomi Young Scholars” program from Xiaomi Foundation.}

\printbibliography

\end{document}